\documentclass[a4paper,twocolumn,11pt,accepted=2025-01-21]{quantumarticle}

\pdfoutput=1

\usepackage[utf8]{inputenc}
\usepackage{amsmath}
\usepackage{amssymb}
\usepackage{amsthm}
\usepackage{physics}
\usepackage{graphicx}
\usepackage{xcolor}
\usepackage{hyperref}
\usepackage{xspace}
\usepackage{bbm}

\usepackage{xstring} 
\usepackage{mathtools} 
\usepackage{bm}
\usepackage{dsfont} 
\usepackage{textgreek} 
\usepackage[group-separator={,}, group-minimum-digits=3]{siunitx} 
\usepackage{booktabs} 
\usepackage[ruled,linesnumbered]{algorithm2e}

\theoremstyle{definition}

\theoremstyle{definition}
\newtheorem{definition}{Definition}

\theoremstyle{definition}
\newtheorem{claim}{Claim}

\usepackage{commands}

\usepackage[numbers,sort,compress]{natbib}

\begin{document}

\title{Color code decoder with improved scaling for correcting circuit-level noise}

\author{Seok-Hyung Lee}
\email{seokhyung.lee@sydney.edu.au}
\affiliation{Centre for Engineered Quantum Systems, School of Physics, The University of Sydney, Sydney, NSW 2006, Australia}
\author{Andrew Li}
\affiliation{Centre for Engineered Quantum Systems, School of Physics, The University of Sydney, Sydney, NSW 2006, Australia}
\author{Stephen D. Bartlett}
\email{stephen.bartlett@sydney.edu.au}
\affiliation{Centre for Engineered Quantum Systems, School of Physics, The University of Sydney, Sydney, NSW 2006, Australia}

\begin{abstract}

Two-dimensional color codes are a promising candidate for fault-tolerant quantum computing, as they have high encoding rates, transversal implementation of logical Clifford gates, and resource-efficient magic state preparation schemes.
However, decoding color codes presents a significant challenge due to their structure, where elementary errors violate three checks instead of just two (a key feature in surface code decoding), and the complexity of extracting syndrome is greater.
We introduce an efficient color-code decoder that tackles these issues by combining two matching decoders for each color, generalized to handle circuit-level noise by employing detector error models.
We provide comprehensive analyses of the decoder, covering its threshold and sub-threshold scaling both for bit-flip noise with ideal measurements and for circuit-level noise.
Our simulations reveal that this decoding strategy nearly reaches the best possible scaling of logical failure ($p_\mr{fail} \sim p^{d/2}$) for both noise models, where $p$ is the noise strength, in the regime of interest for fault-tolerant quantum computing.
While its noise thresholds are comparable with other matching-based decoders for color codes ($8.2\%$ for bit-flip noise and $0.46\%$ for circuit-level noise), the scaling of logical failure rates below threshold significantly outperforms the best matching-based decoders.

\end{abstract}

\maketitle

\section{Introduction}

Two-dimensional (2D) color codes \cite{bombin2006topological,bombin2013topological} are a family of stabilizer quantum error-correcting codes that can be realized with local interactions on a 2D plane, and provide a promising pathway to implementing fault-tolerant quantum computation.  
Compared to surface codes \cite{bravyi1998quantum,dennis2002topological}, color codes have several noteworthy advantages: (i) They have higher encoding rates for the same code distance \cite{landahl2011faulttolerant}, (ii) all the Clifford gates can be implemented transversally \cite{bombin2006topological}, and (iii) an arbitrary pair of commuting logical Pauli product operators can be measured in parallel via lattice surgery \cite{thomsen2024lowoverhead}.
Notably, the transversality of Clifford gates enables highly resource-efficient magic state preparation schemes without distillation \cite{chamberland2020very,itogawa2024even,gidney2024magic}, which can be concatenated with distillation to further improve the quality of the output magic state \cite{hirano2024leveraging,lee2024low}.

Recent experiments have demonstrated significant advancements in the implementation of color codes.
The distance-3 Steane code, a well-studied small instance of a color code, has been used to perform various fault-tolerant operations (such as transversal logic gates, magic state preparation, and lattice surgery) on trapped-ion \cite{postler2022demonstration,ryananderson2022implementing,postler2024demonstration}, neutral atom \cite{bluvstein2024logical,rodriguez2024experimental}, and superconducting \cite{lacroix2024scaling} platforms.
Furthermore, Refs.~\cite{rodriguez2024experimental,lacroix2024scaling} also implement the distance-5 color code, with Ref.~\cite{lacroix2024scaling} demonstrating its improved logical error suppression compared to the Steane code.

For these desirable features to be exploited for fault-tolerant quantum computing in practice, we need better decoders for color codes.  
Surface codes benefit from the many advantages of a decoding approach based on `matching', which is a standard method to handle errors in codes that can only have \emph{edge-like} errors (namely, elementary errors violate at most two checks).  
Matching-based decoders can operate both efficiently and near-optimally, and are readily adapted to handle noisy syndrome extraction circuits. 
In color codes, an elementary error, which is a single-qubit $X$ or $Z$ error, is generally involved in three checks (or stabilizer generators), and so a matching decoder cannot directly be used.
Moreover, considering realistic circuit-level noise makes decoding more difficult because the color code syndrome extraction circuits are more complex than for the surface code. 
As a result of these deficiencies, existing decoders for the color code do not perform as well as expected either in terms of error thresholds or for sub-threshold scaling of the logical failure rate.

Several approaches to decode errors in color codes have been proposed.  The most widely studied methods are the projection decoder and its variants \cite{delfosse2014decoding,chamberland2020triangular,beverland2021cost,kubica2023efficient,zhang2024facilitating}, which apply minimum-weight perfect matching (MWPM) on two or three restricted lattices (which are specific sub-lattices of the dual lattice of the color code) and then deduce the final correction by a procedure called `lifting'.
These approaches can achieve thresholds of around 8.7\% for bit-flip noise \cite{delfosse2014decoding} and around 0.47\% for circuit-level noise \cite{zhang2024facilitating}.
However, they have a fundamental limitation that the logical failure rate $p_\mr{fail}$ scales like $p^{d/3}$ below threshold \cite{beverland2021cost,sahay2022decoder,zhang2024facilitating}, not $p^{d/2}$, where $p$ is a physical noise strength and $d$ is the code distance.
This is because there exist errors with weights $O(d/3)$ that are uncorrectable via the decoder (noting that, in principle, an optimal decoder can correct any error with weight up to $d/2$).
Such a drawback may significantly hinder resource efficiency of quantum computing, necessitating a larger code distance to maintain the same $p_\mr{fail}$, compared to a scenario using a decoder with the expected optimal scaling of $p^{d/2}$.
Roughly estimating based on the ansatz $p_\mr{fail} = \alpha (p/p_\mathrm{th})^{\beta d + \eta}$ for the threshold $p_\mathrm{th}$ and parameters $\alpha$, $\beta$, and $\eta$, using a decoder with $\beta = 1/3$ demands about $9/4$ times more qubits than an optimal decoder with $\beta=1/2$ to achieve a similar $p_\mr{fail}$, assuming $p_\mathrm{th}$, $\alpha$, and $\eta$ remain similar in both cases.

The Möbius MWPM decoder \cite{sahay2022decoder} is another matching-based decoder implemented by applying the MWPM algorithm on a manifold built by connecting the three restricted lattices, which has the topology of a Möbius strip.
It achieves a higher threshold of 9.0\% under bit-flip noise and, more importantly, a better scaling of $p_\mr{fail} \sim p^{3d/7}$.
The decoder has subsequently been improved to accommodate circuit-level noise and general color-code lattices \cite{gidney2023new}.

Besides these matching-based decoders, there are tensor network decoders (achieving a threshold of 10.9\% under bit-flip noise) \cite{chubb2021general}, simulated annealing (10.4\%) \cite{takada2024ising}, MaxSAT problem (10.1\%) \cite{berent2024decoding}, trellis (10.1\%) \cite{sabo2022trellis}, neural network (10.0\%) \cite{maskara2019advantages}, union-find (8.4\%) \cite{delfosse2021almostlinear}, and renormalization group (7.8\%) \cite{sarvepalli2012efficient}.  
All of these reported thresholds are for the 6-6-6 (or hexagonal) color-code lattice except that of the decoder based on simulated annealing \cite{takada2024ising}, which is for the 4-8-8 lattice.
However, it is currently unclear how these decoders can be adapted to circuit-level noise and how their performance will be.
We note that the tensor network decoder can be adapted to treat noisy syndrome measurements in a phenomenological or circuit-level noise model, but it is highly inefficient due to the difficulty of 3D tensor network contraction \cite{piveteau2024tensor}.

In this work, we propose a matching-based color-code decoder, which we call the \emph{concatenated MWPM decoder}, that demonstrates exceptional sub-threshold scaling of the logical failure rate in regimes of interest for fault-tolerant quantum computing.  Our decoder functions by `concatenation' of two MWPM decoders per color, for a total of six matchings.
Roughly speaking, this decoder is based on the idea that decoding syndrome on a specific (say, red) restricted lattice returns a prediction of red edges with odd-parity errors, which can be combined with the syndrome data of red checks (specifying red faces with odd-parity errors) and decoded again to predict errors.
This process is repeated for each of the three colors and the most probable prediction is selected as the final prediction.
We demonstrate that this decoder can be generalized to handle circuit-level noise by using the concept of \emph{detector error model}.
This approach to color-code decoding was inspired by a decoding approach for measurement-based quantum computing~\cite{lee2022universal}, and here we develop and generalize it further for gate-based quantum computing.

We analyze the performance of the concatenated MWPM decoder against bit-flip and circuit-level noise models by evaluating its noise thresholds and investigating the sub-threshold scaling of logical failure rates.
Notably, the logical failure rates of the concatenated MWPM decoder below threshold is shown to be well-described by the scaling $p_\mr{fail} \sim p^{d/2}$ for both of the noise models within the range of our simulations ($d \leq 31$ for bit-flip noise and $d \leq 21$ for circuit-level noise).
Thanks to this improvement, although it has comparable or slightly lower thresholds (8.2\% for bit-flip noise and 0.46\% for circuit-level noise), its sub-threshold performance significantly surpasses the projection decoder.
Compared to the Möbius decoder, our decoder has a similar scaling factor against $d$ but achieves approximately 3--7 times lower logical failure rates for circuit-level noise when $10^{-4} \lessapprox p \lessapprox 5 \times 10^{-4}$.

A python module for simulating color code circuits and running the concatenated MWPM decoder is publicly available on Github \cite{colorcodestim}.

This paper is structured as follows:
In Sec.~\ref{sec:preliminaries}, we provide a brief overview of color codes and their decoding problem, and describe our methodology of analysis including the settings, noise models, and criteria for evaluating decoder performance.
In Sec.~\ref{sec:decoder_bit_flip}, we introduce the 2D variant of the concatenated MWPM decoder that works when syndrome measurements are perfect and analyze it numerically for bit-flip noise.
In Sec.~\ref{sec:decoder_circuit_level}, we generalize the decoder by using detector error models to accommodate circuit-level noise including faulty syndrome measurements and present the outcomes of its numerical analyses as well.
We conclude with final remarks in Sec.~\ref{sec:remarks}.

\section{Preliminaries \label{sec:preliminaries}}

\subsection{Color codes}

\begin{figure}[!t]
    \centering
    \includegraphics[width=\linewidth]{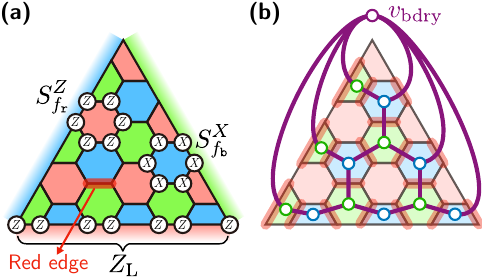}
    \caption{
        \textbf{Example of a triangular color code and its restricted lattice.}
        \subfig{a} The triangular color code with code distance $d=7$ based on a hexagonal lattice with three boundaries of different colors is displayed.
        Each face is associated with a pair of $Z$-type and $X$-type checks, exemplified by $\sgz{f_\rbs}$ on a red face and $\sgx{f_\bbs}$ on a blue face, respectively. 
        Logical operators $X_L$ and $Z_L$ can be supported on one of the three boundaries. 
        Each edge is colored by the color of the faces connected by it, as exemplified by the red edge highlighted.
        \subfig{b} The red-restricted lattice of the code is illustrated, where vertices and edges are shown as circles and purple lines, respectively.
        Its vertices consist of blue faces (blue circles), green faces (green circles), and an additional boundary vertex $v_\mr{bdry}$ (purple circle).
        Each of its edges corresponds to a red edge of the original lattice, which are highlighted as thick red lines.
    }
    \label{fig:color_code}
\end{figure}

A 2D color-code lattice indicates a lattice satisfying the following two conditions:
\begin{itemize}
    \item \textit{3-valent}: Each vertex is connected with three edges.
    \item \textit{3-colorable}: One of three colors can be assigned to each face in a way that adjacent faces do not have the same color.
\end{itemize}
Following the usual convention, we use red (\rbs), green (\gbs), and blue (\bbs) as these three colors assigned to faces.
Note that each edge is also colorable by the color of the faces that it connects.
There are various types of color-code lattices, among which hexagonal (or 6-6-6) lattices are used to test the concatenated MWPM decoder in this work.
Nonetheless, since our decoder is described in a form that does not depend on the lattice type, it can also be applied to other color-code lattices such as 4-8-8 lattices \cite{fowler2011twodimensional}.

Given a 2D color-code lattice, a 2D color code \cite{bombin2006topological} is defined on qubits placed at the vertices of the lattice.
The code space of the color code is stabilized by two types of \emph{checks} $\sgx{f}$ and $\sgz{f}$ for each face $f$, which are defined as
\begin{align}
    \sgx{f} \coloneqq \prod_{v \in f} \Xv{v}, \qquad \sgz{f} \coloneqq \prod_{v \in f} \Zv{v}, \label{eq:checks_def}
\end{align}
where $\Xv{v}$ and $\Zv{v}$ are the Pauli-X and Z operators on the qubit placed at $v$ and the notation `$v \in f$' means that the vertex $v$ is included in $f$.
In other words, the code space is the common $+1$ eigenspace of these operators.
We categorize checks according to their Pauli types and the colors of the corresponding faces; for example, if $f$ is a red face, $\sgx{f}$ ($\sgz{f}$) is a red $X$-type ($Z$-type) check.

To define a single logical qubit with a color code, we can use a triangular patch possessing three boundaries of different colors, as shown in Fig.~\ref{fig:color_code}(a) for the code distance of $d=7$ with several examples of its checks and logical-$Z$ operator.
Here, we say that a boundary has a specific color (e.g., red) if and only if it is adjacent to only faces of other two colors (e.g., green and blue).
In addition, an edge connecting a boundary and a face is regarded to have the same color as them.
The logical Pauli-$X$ ($Z$) operator, denoted as $\ov{X}$ ($\ov{Z}$), is defined as the product of $X$ ($Z$) operators on qubits placed along one of the three boundaries.
The weight of each of these logical Pauli operators is the code distance of the code.

\subsection{Decoding problem}

For detecting and correcting errors, the checks of Eq.~\eqref{eq:checks_def} are measured, repeatedly.
These measurements consist of multiple \textit{rounds}, each round being a subroutine to measure all the commuting checks in parallel without duplication (see Sec.~\ref{subsec:circuit} for the circuit implementing this).
Each measurement outcome of a check is referred to as a \emph{check outcome} and, if it is $-1$, we say that the check is \emph{violated}.
The collection of check outcomes of a single round is called a \emph{syndrome}.

\emph{Decoding} is the process of predicting errors from syndromes of a single or multiple rounds.
If the residual errors after decoding incurs a logical error, we say that the decoding fails.
Ideal decoding using maximum likelihood is generally known to be \textsf{\#P-complete} \cite{iyer2015hardness}, thus alternative decoders that are less precise but more efficient are necessary for practical quantum computing.
For Calderbank-Shor-Steane (CSS) codes where $X$-type and $Z$-type checks are separately defined, these two types of checks can be decoded independently by regarding $Y$ errors as combinations of $X$ and $Z$ errors.

Assuming that syndrome measurements are perfect, meaning each check outcome is obtained without error, the minimum-weight perfect matching (MWPM) algorithm \cite{edmonds1965paths} can be used directly for decoding Pauli errors on data qubits when each of elementary Pauli errors (which typically consist of the Pauli-$X$ and $Z$ operators of all data qubits) anticommutes with at most two checks of the code. 
This is the situation for surface codes \cite{bravyi1998quantum}.
In such a case, we consider a \emph{matching graph} $G_\mr{mat} = (V, E)$ whose vertices $V$ consist of checks and an additional `boundary vertex' $v_\mr{bdry}$.
For each elementary error $P$, two checks anticommuting with $P$ (or $v_\mr{bdry}$ and a single check anticommuting with $P$) are connected within $G_\mr{mat}$.
A weight is assigned to each edge either uniformly or as $\log[(1-q)/q]$, where $q$ is the probability of the corresponding error.
Denoting the set of violated checks as $\sigma$, the MWPM algorithm identifies a minimal-total-weight set of edges, denoted as $\MWPM{G_\mr{mat}}{v_\mr{bdry}}{\sigma}$, that meet each violated check an odd number of times and each unviolated check an even number of times.
In other words, $E_\mr{MWPM} \coloneqq \MWPM{G_\mr{mat}}{v_\mr{bdry}}{\sigma} \subset E$ satisfies
\begin{align*}
    &\forall v \in \sigma, \; \abs{\qty{e \in E_\mr{MWPM} \mid v \in e}} \equiv 1 \pmod{2}, \\ 
    &\forall v \in V \setminus \qty(\sigma \cup \qty{v_\mr{bdry}}), \; \\
    &\qquad \abs{\qty{e \in E_\mr{MWPM} \mid v \in e}} \equiv 0 \pmod{2}
\end{align*}
and have the smallest sum of weights.
The final prediction is then the set of elementary errors corresponding to the edges in $E_\mr{MWPM}$.

However, the above method is not applicable to color codes since each of the Pauli-$X$ and $Z$ operators of their physical qubits can anticommute with three checks.
To obviate this problem, we can instead consider three \emph{restricted lattices} visualized in Fig.~\ref{fig:color_code}(b), which are derived from the color-code lattice $\LatticeTwoD$ as follows:
The red-restricted lattice $\LatticeRest{\rbs}$ contains the green and blue faces and an additional `boundary vertex' $v_\mr{bdry}$ as its vertices and, for each red edge $e^{(\rbs)}$ of $\LatticeTwoD$, the green and blue faces divided by $e^{(\rbs)}$ (or $v_\mr{bdry}$ and $f$ if $e^{(\rbs)}$ belongs to only one face $f$) are connected within $\LatticeRest{\rbs}$ by an edge.\footnote{In the literature, the restricted lattice is more often defined to have two boundary vertices (connected with each other), which respectively correspond to two among the three boundaries except the boundary of the restricted color. Although it may be mathematically more natural, it is not necessary for MWPM since the edge between the two boundary vertices is regarded to have zero weight. See Appendix~\ref{subapp:boundary_consideration} for more details.}
The green- and blue-restricted lattices $\LatticeRest{\gbs}$, $\LatticeRest{\bbs}$ are defined analogously.
Importantly, any single-qubit $X$ and $Z$ errors respectively affect checks on at most two vertices in each restricted lattice.
Therefore, we can apply MWPM on (all or some of) these three restricted lattices individually and combine these outcomes through a specific method, which is the basic idea of the projection decoder \cite{delfosse2014decoding}.
Alternatively, MWPM can be performed on a manifold built by connecting the three restricted lattices appropriately as in the Möbius MWPM decoder \cite{sahay2022decoder}.
In this work, we take another approach to perform MWPM twice in a concatenated way, where the first one is applied on a restricted lattice but the second one is applied on another useful lattice derived from $\LatticeTwoD$ in a different way.

We note that we have assumed perfect syndrome extraction in the above discussions.
If syndrome extraction is noisy, we can generally employ detector error models \cite{gidney2021stim}, which will be elaborated on in Sec.~\ref{sec:decoder_circuit_level}.

\subsection{Noise models and settings \label{subsec:settings}}

We consider one of the two types of noise models: bit-flip and circuit-level noise models.
In the bit-flip noise model of strength $p$, every data qubit undergoes an $X$ error with probability $p$ at the start of each round.
There are no errors in other steps including initialization and measurement.
A more sophisticated noise model is the circuit-level noise model of strength $p$, defined as follows:
\begin{itemize}
    \item Every measurement outcome is flipped with probability $p$.
    \item Every preparation of a qubit produces an orthogonal state with probability $p$.
    \item Every single- or two-qubit unitary gate (including the idle gate $I$) is followed by a single- or two-qubit depolarizing noise channel of strength $p$. We here regard that, for every time slice of the circuit, idle gates $I$ are acted on all the qubits that are not involved in any non-trivial unitary gates or measurements.
\end{itemize}
Here, the single- and two-qubit depolarizing channels of strength $p$ are respectively defined as
\begin{align*}
    &\depchannelone{p}:\; \rho^{(1)} \mapsto (1 - p) \rho^{(1)} + \frac{p}{3} \sum_{P \in \qty{X, Y, Z}} P \rho^{(1)} P, \\ 
    &\depchanneltwo{p}:\; \rho^{(2)} \mapsto (1 - p) \rho^{(2)} + \frac{p}{15} \\ 
    &\qquad\quad \times \sum_{\substack{P_1, P_2 \in \qty{I, X, Y, Z} \\ P_1 \otimes P_2 \neq I \otimes I}} \qty(P_1 \otimes P_2) \rho^{(2)} \qty(P_1 \otimes P_2),
\end{align*}
where $\rho^{(1)}$ and $\rho^{(2)}$ are arbitrary single- and two-qubit density matrices, respectively.

To evaluate the performance of our decoder, we simulate $T$ rounds of the logical idling gate acting on the triangular color code of code distance $d$, which is preceded by logical initialization (to the $+1$ eigenstate $\ket{\ov{0}}$ of $\ov{Z}$) and followed by a measurement in the $\ov{Z}$ basis.
The logical initialization and measurement are done by initializing all the data qubits to $\ket{0}$ and measuring them in the $Z$ basis, respectively.
From this scenario, we can identify whether the $\ov{Z}$ observable fails (i.e., a $\ov{X}$ or $\ov{Y}$ error occurs) during the idling gate.
While it is sufficient for bit-flip noise models, the $\ov{X}$ observable also can fail under circuit-level noise models.
To cover this, we simulate again by swapping the $X$- and $Z$-type parts from the \cnot schedule (see Sec.~\ref{subsec:optimizing_cnot_schedule} for more details).
We estimate the logical failure rate $p_\mr{fail}$ as $p_\mr{fail} = p_\mr{fail}^{(X)} + p_\mr{fail}^{(Z)}$, where $p_\mr{fail}^{(X)}$ and $p_\mr{fail}^{(Z)}$ are respectively the failure rates of the $\ov{X}$ and $\ov{Z}$ observables.
Since a $\ov{Y}$ error causes both observables to fail, this provides a conservative estimate that upper-bounds the actual failure rate.

\subsection{Assessment of decoder performance}

We assess the performance of the decoder from two perspectives: (i) what the noise threshold $\CrossThrs{}$ of the decoder is and (ii) how the logical failure rate behaves when the noise strength $p$ is sufficiently lower than the threshold (namely, $p \lessapprox \CrossThrs{}/2$).

For a given value of $T$, the (cross) noise threshold $\CrossThrs{}(T)$ is defined by the noise strength at which the curves of the logical failure rates $p_\mr{fail}$ for different code distances (which are sufficiently large) cross each other.
We also define the \emph{long-term cross threshold} as $\CrossThrsLT{} \coloneqq \lim_{T \rightarrow \infty} \CrossThrs{}(T)$ and predict it by fitting data into the ansatz
\begin{align}
    \CrossThrs{}(T) = \CrossThrsLT{} \qty[1 - \qty(1 - \frac{\CrossThrs{}(1)}{\CrossThrsLT{}}) T^{-\gamma}],
    \label{eq:long_term_thrs_ansatz}
\end{align}
following the method used in Ref.~\cite{beverland2021cost}.

\begin{figure*}[!t]
    \centering
    \includegraphics[width=\linewidth]{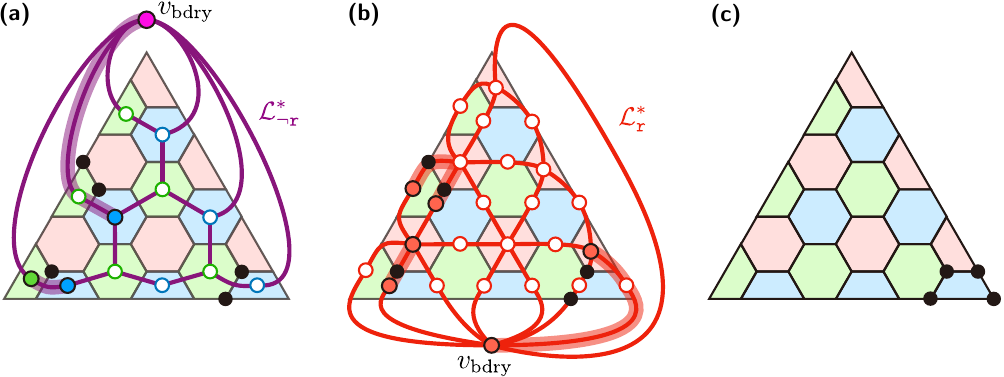}
    \caption{
        \textbf{Example of executing the concatenated MWPM decoder on the triangular color code of distance 7.}
        \subfig{a} First-round MWPM is performed on the red-restricted lattice $\LatticeRest{\rbs}$, where the vertices are the green and blue faces of the original lattice $\LatticeTwoD$ (marked as green and blue empty/filled circles) and a boundary vertex $v_\mr{bdry}$ (marked as a purple filled circle) connected by purple solid lines.
        The green/blue faces with violated checks are highlighted as green/blue filled circles.
        The obtained matching is drawn as purple thick lines, which correspond to a set of red edges $E_\mr{pred}^{(\rbs)}$.
        \subfig{b} Second-round MWPM is performed on the red-only lattice $\LatticeMono{\rbs}$, where the vertices are the red faces and edges of $\LatticeTwoD$ and a boundary vertex $v_\mr{bdry}$ (marked as red empty/filled circles) connected by red lines.
        The red edges in $E_\mr{pred}^{(\rbs)}$ and the red faces with violated checks are highlighted as red filled circles.
        The obtained matching is drawn as red thick lines, which corresponds to a set of vertices $V_\mr{pred}^{(\rbs)}$ predicted to have errors.
        \subfig{c} Residual errors after correcting $V_\mr{pred}^{(\rbs)}$ are marked as black dots, which is equivalent to a stabilizer.
        The above process is repeated two more times while varying colors and the smallest one among $V_\mr{pred}^{(\rbs)}$, $V_\mr{pred}^{(\gbs)}$, and $V_\mr{pred}^{(\bbs)}$ is selected as the outcome.
    }
    \label{fig:decoder_2D}
\end{figure*}

To investigate the sub-threshold scaling, we fit data (with sufficiently low $p$) into the ansatz\footnote{
    This ansatz, which is widely employed in the literature \cite{fowler2013surface,fowler2013analytic,fowler2019low,litinski2019game,litinski2019magic,sahay2022decoder}, can be justified as follows:
    For sufficiently low $p$, $p_\mr{fail}$ can be approximated by leaving only its lowest-order term with respect to $p$ (with no constant term, since $p_\mr{fail} = 0$ when $p = 0$). Furthermore, below the threshold, $p_\mr{fail}$ is expected to decrease exponentially with $d$, as proven for toric and surface codes \cite{dennis2002topological,raussendorf2007topological,fowler2012proof,watson2014logical} and supported numerically for color codes \cite{beverland2021cost,sahay2022decoder}.
    Hence, for sufficiently low $p$, it is natural to consider an ansatz where $\log p_\mr{fail}$ is bilinear in $\log p$ and $d$, which leads to the form of Eq.~\eqref{eq:pfail_ansatz}.
}
\begin{align}
    \frac{p_\mr{fail}}{T} = \alpha\qty(\frac{p}{\ScalingThrs{}})^{\beta \qty(d - d_0) + \eta}
    \label{eq:pfail_ansatz}
\end{align}
with five parameters $\ScalingThrs{}$, $\alpha$, $\beta$, $\eta$ and $d_0$, where $d_0$ is a number determined appropriately to minimize the uncertainties of $\alpha$ and $\eta$.
Equation~\eqref{eq:pfail_ansatz} can be equivalently written as
\begin{align}
    \log \qty(\frac{p_\mr{fail}}{T}) = G(d) \log p + C(d), \label{eq:log_pfail_ansatz}
\end{align}
where
\begin{align*}
    G(d) &= \beta (d - d_0) + \eta, \\ 
    C(d) &= - (\beta \log \ScalingThrs{})(d - d_0) + \qty(\log \alpha - \eta \log \ScalingThrs{}).
\end{align*}
Hence, we can determine these parameters by two steps of linear regressions: first computing $G(d)$ and $C(d)$ by fitting $\log (p_\mr{fail}/T)$ against $\log p$ for each $d$ and then fitting $G(d)$ and $C(d)$ against $d$, respectively.
We select $d_0$ such that the uncertainties of the constant terms (i.e., $\eta$ and $\log\alpha - \eta\log\ScalingThrs{}$) of $G(d)$ and $C(d)$ are minimized.
(If such values of $d_0$ differ for $G(d)$ and $C(d)$, select their average.)

Since $p_\mr{fail}/T$ is expected to be constant when $p = \ScalingThrs{}$, we regard $\ScalingThrs{}$ as a `noise threshold' as well.
To distinguish the two types of thresholds $\CrossThrs{}$ and $\ScalingThrs{}$, we refer to them as the \emph{cross threshold} and \emph{scaling threshold}, respectively.
Note that $\CrossThrs{}$ and $\ScalingThrs{}$ are generally not equal, as the ansatz in Eq.~\eqref{eq:pfail_ansatz} is an asymptotical expression valid for sufficiently small $p$, where the sub-leading order terms with respect to $p$ are negligible.
Additionally, we specify the noise model (bit-flip or circuit-level) on the subscripts of the thresholds such as $\CrossThrs{bitflip}$, $\CrossThrs{circuit}$, $\ScalingThrs{bitflip}$, and $\ScalingThrs{circuit}$.

\section{Concatenated MWPM decoder without faulty measurements \label{sec:decoder_bit_flip}}

In this section, we describe the 2D variant of the concatenated MWPM decoder that is applicable only when every syndrome measurement is perfect.
Let us consider the 2D color code on a lattice $\LatticeTwoD$, which may have boundaries.
The decoder to predict $X$ errors in a single round can be briefly depicted as follows (see Fig.~\ref{fig:decoder_2D} for an example):
\begin{enumerate}
    \item \textbf{(First-round MWPM)} Input violated blue and green $Z$-type checks to the MWPM algorithm on the red-restricted lattice $\LatticeRest{\rbs}$, which returns a set of edges of $\LatticeRest{\rbs}$. 
    This set corresponds to a set $E_\mr{pred}^{(\rbs)}$ of red edges of $\LatticeTwoD$, each of which is predicted to contain one $X$ error. 
    See Fig.~\ref{fig:decoder_2D}(a).
    \item \textbf{(Second-round MWPM)} Input $E_\mr{pred}^{(\rbs)}$ and violated red $Z$-type checks to the MWPM algorithm on the `red-only lattice' $\LatticeMono{\rbs}$, which is constructed according to the connection structure of red edges and faces in $\LatticeTwoD$.
    The algorithm returns a set of edges of $\LatticeMono{\rbs}$, which correspond to a set of vertices $V_\mr{pred}^{(\rbs)}$ of $\LatticeTwoD$ that are predicted to have errors.
    See Fig.~\ref{fig:decoder_2D}(b).
    \item Repeat the above two steps (together referred to as the \textit{red sub-decoding procedure}) while varying the color to green and blue, obtaining $V_\mr{pred}^{(\gbs)}$ and $V_\mr{pred}^{(\bbs)}$. Select the smallest one $V_\mr{pred}$ among $V_\mr{pred}^{(\rbs)}$, $V_\mr{pred}^{(\gbs)}$, and $V_\mr{pred}^{(\bbs)}$ as the final outcome.
\end{enumerate}
Figure~\ref{fig:decoder_2D}(c) presents a set of residual errors after correction, which is equivalent to a stabilizer thus does not cause a logical failure.
$Z$ errors can be predicted by decoding $X$-type check outcomes in an analogous way.

To formally describe the above procedure, let us first define some notations.
For a lattice $\mathcal{L}$, the sets of its vertices, edges, and faces are respectively denoted as $\latelm{\mathcal{L}}{0}$, $\latelm{\mathcal{L}}{1}$, and $\latelm{\mathcal{L}}{2}$.
For each color $\cbs \in \qty{\rbs, \gbs, \bbs}$, we denote the sets of \cbs-colored edges and faces in $\LatticeTwoD$ as $\latelmcolor{\LatticeTwoD}{1}{\cbs}$ and $\latelmcolor{\LatticeTwoD}{2}{\cbs}$.
The set of faces with violated $Z$-type checks is denoted as $\sigma_Z \in \latelm{\LatticeTwoD}{2}$.
We partition $\sigma_Z$ into $\sigma_Z = \sigma_Z^{(\rbs)} \cup \sigma_Z^{(\gbs)} \cup \sigma_Z^{(\bbs)}$ such that $\sigma_Z^{(\cbs)} \subseteq \latelmcolor{\LatticeTwoD}{2}{\cbs}$ for each $\cbs \in \qty{\rbs, \gbs, \bbs}$.
The restricted and monochromatic lattices are then formally defined as follows:
\begin{definition}[\textbf{Restricted lattices}]
    The \emph{red-restricted lattice} $\LatticeRest{\rbs}$ is defined to have vertices of $\latelm{\LatticeRest{\rbs}}{0} = \latelmcolor{\LatticeTwoD}{2}{\gbs} \cup \latelmcolor{\LatticeTwoD}{2}{\bbs} \cup \qty{v_\mr{bdry}}$, where $v_\mr{bdry}$ is an additional boundary vertex. 
    For each red edge $e \in \latelmcolor{\LatticeTwoD}{1}{\rbs}$, we create an edge denoted as $\EdgeCorrRest{\rbs}(e)$ within $\LatticeRest{\rbs}$ as follows:
    If $e$ belongs to two faces $f_\gbs \in \latelmcolor{\LatticeTwoD}{2}{\gbs}$ and $f_\bbs \in \latelmcolor{\LatticeTwoD}{2}{\bbs}$, $\EdgeCorrRest{\rbs}(e)$ connects $f_\gbs$ and $f_\bbs$.
    If $e$ belongs to only one face $f \in \latelmcolor{\LatticeTwoD}{2}{\gbs} \cup \latelmcolor{\LatticeTwoD}{2}{\bbs}$, $\EdgeCorrRest{\rbs}(e)$ connects $f$ and $v_\mr{bdry}$.
    Note that $\EdgeCorrRest{\rbs}$ is a bijection between $\latelmcolor{\LatticeTwoD}{1}{\rbs}$ and $\latelm{\LatticeRest{\rbs}}{1}$.
    For $E \subseteq \latelmcolor{\LatticeTwoD}{1}{\rbs}$, we denote $\EdgeCorrRest{\rbs}(E) \coloneqq \qty{ \EdgeCorrRest{\rbs}(e) \mid e \in E}$.
    The green- and blue-restricted lattices $\LatticeRest{\gbs}$, $\LatticeRest{\bbs}$ (with bijections $\EdgeCorrRest{\gbs}$ and $\EdgeCorrRest{\bbs}$) are defined similarly.
    \label{def:restricted_lattices}
\end{definition}
\begin{definition}[\textbf{Monochromatic lattices}]
    The \emph{red-only lattice} $\LatticeMono{\rbs}$ is defined to have vertices of $\latelm{\LatticeMono{\rbs}}{0} = \latelmcolor{\LatticeTwoD}{1}{\rbs} \cup \latelmcolor{\LatticeTwoD}{2}{\rbs} \cup \qty{v_\mr{bdry}}$, where $v_\mr{bdry}$ is an additional boundary vertex. 
    For each vertex $v \in \latelm{\LatticeTwoD}{0}$, we connect an edge denoted as $\EdgeCorrMono{\rbs}(v)$ within $\LatticeMono{\rbs}$ as follows:
    If $v$ is commonly included in a red edge $e \in \latelmcolor{\LatticeTwoD}{1}{\rbs}$ and a red face $f \in \latelmcolor{\LatticeTwoD}{2}{\rbs}$, $\EdgeCorrMono{\rbs}(v)$ connects $e$ and $f$.
    If $v$ is only included either in a red face $f$ or in a red edge $e$, $\EdgeCorrMono{\rbs}(v)$ connects $f$ or $e$ with $v_\mr{bdry}$.
    Note that $\EdgeCorrMono{\rbs}$ is a bijection between $\latelm{\LatticeTwoD}{0}$ and $\latelm{\LatticeMono{\rbs}}{1}$.
    For $V \subseteq \latelm{\LatticeTwoD}{0}$, we denote $\EdgeCorrMono{\rbs}(V) \coloneqq \qty{ \EdgeCorrMono{\rbs}(v) \mid v \in V}$.
    The green- and blue-only lattices $\LatticeMono{\gbs}$, $\LatticeMono{\bbs}$ (with bijections $\EdgeCorrMono{\gbs}$ and $\EdgeCorrMono{\bbs}$) are defined similarly.
    \label{def:monochromatic_lattices}
\end{definition}

Using the above definitions, we formally describe the decoding process as follows:
\begin{enumerate}
    \item For each $\cbs \in \qty{\rbs, \gbs, \bbs}$, execute the \cbs-colored sub-decoding procedure as follows:
    \begin{enumerate}
        \item Perform the MWPM algorithm to identify a set $\wtilde{E}_\mr{pred}^{(\cbs)} \coloneqq \MWPM{\LatticeRest{\cbs}}{v_\mr{bdry}}{\sigma_Z^{(\cbs_1)} \cup \sigma_Z^{(\cbs_2)}} \subseteq \latelm{\LatticeRest{\cbs}}{1}$, where $\cbs_1$ and $\cbs_2$ are the other two colors besides $\cbs$.
        Define $E_\mr{pred}^{(\cbs)} \coloneqq \EdgeCorrRest{\cbs}^{-1} \qty(\wtilde{E}_\mr{pred}^{(\cbs)}) \in \latelmcolor{\LatticeTwoD}{1}{\cbs}$.
        \item Perform the MWPM algorithm to identify a set $\wtilde{V}_\mr{pred}^{(\cbs)} \coloneqq \MWPM{\LatticeMono{\cbs}}{v_\mr{bdry}}{\sigma_Z^{(\cbs)} \cup E_\mr{pred}^{(\cbs)}} \subseteq \latelm{\LatticeMono{\cbs}}{1}$ and define $V_\mr{pred}^{(\cbs)} \coloneqq \EdgeCorrMono{\cbs}^{-1} \qty(\wtilde{V}_\mr{pred}^{(\cbs)}) \in \latelm{\LatticeTwoD}{0}$
    \end{enumerate}
     \item Return the smallest one $V_\mr{pred}$ among $V_\mr{pred}^{(\rbs)}$, $V_\mr{pred}^{(\gbs)}$, and $V_\mr{pred}^{(\bbs)}$.
\end{enumerate}

In Appendix~\ref{app:validity_proof}, we formally prove that the above process always returns a valid error prediction consistent with the syndrome.
Additionally, we confirm that any outcome obtained from the projection decoder is a valid matching for the concatenated MWPM decoder, which implies that the former cannot outperform the latter.

In Fig.~\ref{fig:color_strategy_comparison} of Appendix~\ref{app:more_numerical_analyses}, we numerically verify that our color-selecting strategy of executing the sub-decoding procedures for all the three colors and selecting the best one is indeed necessary to maximize the performance of the decoder.
If only one or two colors are considered in this process, the failure rate of the decoder significantly increases.

\subsection{Performance analysis \label{subsec:bitflip_performance_analysis}}

To assess the performance of the decoder, we consider the setting in Sec.~\ref{subsec:settings} under the bit-flip noise model of strength $p$.
Since the bit-flip noise is independently applied to each round, the cross threshold is invariant under $T$, thus we simply denote $\CrossThrs{bitflip} \coloneqq \CrossThrs{bitflip}(T) = \CrossThrsLT{bitflip}$ and consider only the case of $T=1$.
We employ the \pymatching library \cite{higgott2022pymatching} to run the MWPM algorithm.

\begin{figure}[!t]
    \centering
    \includegraphics[width=\linewidth]{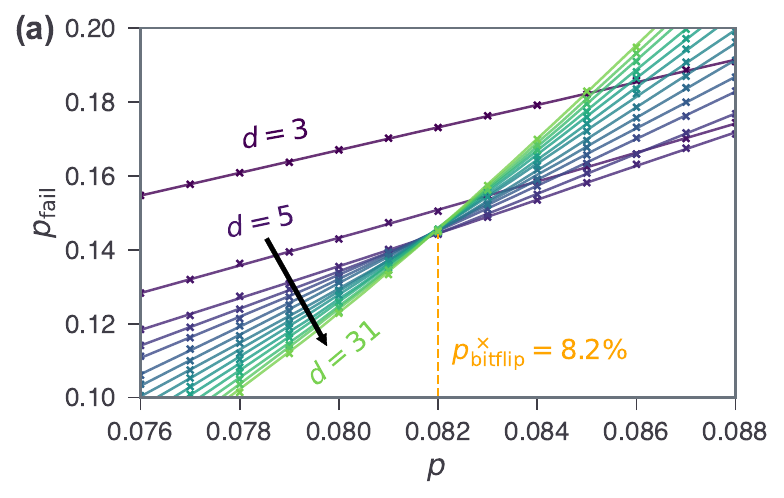}
    \includegraphics[width=\linewidth]{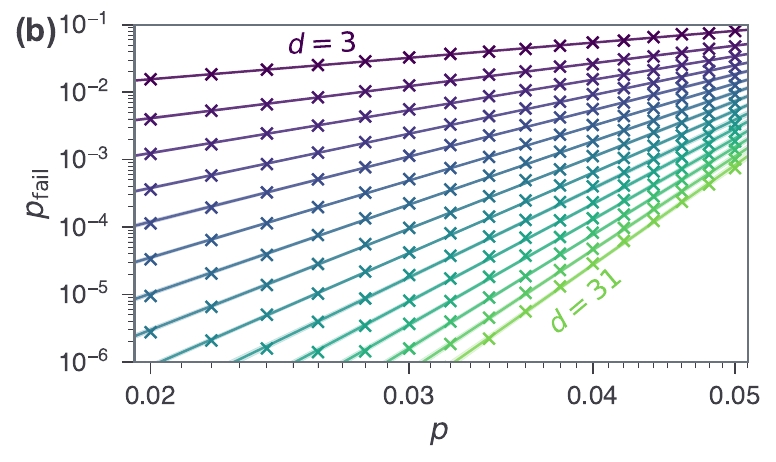}
    \includegraphics[width=\linewidth]{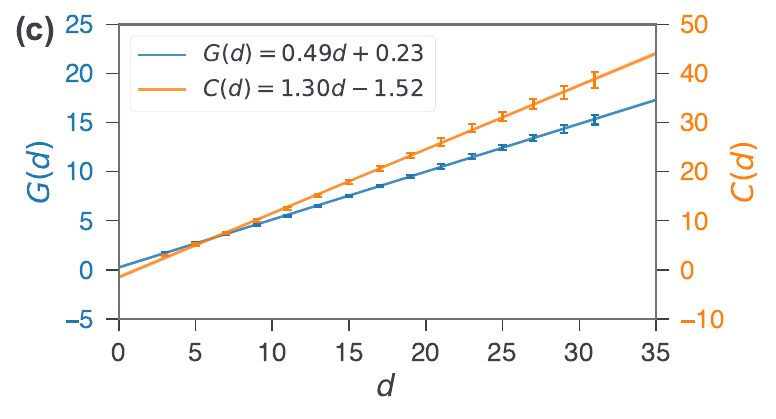}
    \caption{
        \textbf{Numerical analysis of the concatenated MWPM decoder under bit-flip noise.}
        We consider a single round of QEC on a triangular patch with code distance $d$ under the bit-flip noise model of strength $p$.
        The logical failure rates~$p_\mr{fail}$ are plotted over $p$ \subfig{a} near the threshold and \subfig{b} in a sub-threshold region for various code distances ($d=3, 5, 7, \cdots, 31$).
        In \subfig{a}, the solid lines are the LOWESS regressions with fraction $2/3$.
        The obtained cross threshold is $p_\mr{bitflip}^\times = 8.2\%$.
        In \subfig{b}, the solid lines are the linear regressions in the logarithmic scale, each of which is in the form of Eq.~\eqref{eq:log_pfail_ansatz}.
        The 99\% confidence intervals (CIs) of the regression estimates are depicted as shaded regions around the solid lines.
        The number of samples for obtaining each data point is selected so that the 99\% CI of $p_\mr{fail}$ is $\pm 10^{-3}$ for \subfig{a} and $\pm 0.05p_\mr{fail}$ for \subfig{b}.
        In \subfig{c}, the 99\% CIs of $G(d)$ and $C(d)$ in Eq.~\eqref{eq:log_pfail_ansatz} are plotted over $d$ with their linear regressions.
        From the regression parameters, we obtain the parameters of the ansatz of Eq.~\eqref{eq:pfail_ansatz} as presented in Eq.~\eqref{eq:ansatz_parameters_bitflip}.
    }
    \label{fig:bitflip_numerical_analysis}
\end{figure}

In Fig.~\ref{fig:bitflip_numerical_analysis}(a) and~(b), we present the logical failure probabilities $p_\mr{fail}$ computed for various code distances $d$ when $p$ is near the threshold and when $p$ is sufficiently lower than the threshold, respectively.
From Fig.~\ref{fig:bitflip_numerical_analysis}(a), we can clearly observe that the cross threshold is $\CrossThrs{bitflip} \approx 8.2\%$.
In Fig.~\ref{fig:bitflip_numerical_analysis}(b), the data points are linearly fitted for each $d$ in the logarithmic scale, where the slopes $G(d)$ and constant terms $C(d)$ are plotted over $d$ in Fig.~\ref{fig:bitflip_numerical_analysis}(c).
The regression parameters of $G(d)$ and $C(d)$ are estimated as
\begin{align*}
    G(d) = (0.488 \pm 0.006) (d - 17) + (8.53 \pm 0.05), \\
    C(d) = (1.30 \pm 0.02) (d - 17) + (20.6 \pm 0.2),
\end{align*}
by selecting $d_0 = 17$, where the error terms are the 99\% confidence intervals (CIs) estimated based on Students' $t$-distribution by using the python module \texttt{statsmodels} \cite{seabold2010statsmodels}.
The corresponding ansatz parameters in Eq.~\eqref{eq:pfail_ansatz} are estimated as
\begin{align}
    \begin{split}
        &\ScalingThrs{bitflip} \approx 0.069, \quad \alpha \approx 0.12, \quad \beta \approx 0.49, \\ 
        &\eta \approx 8.5, \quad d_0 = 17.
    \end{split}
    \label{eq:ansatz_parameters_bitflip}
\end{align}


\begin{figure}[!t]
    \centering
    \includegraphics[width=\linewidth]{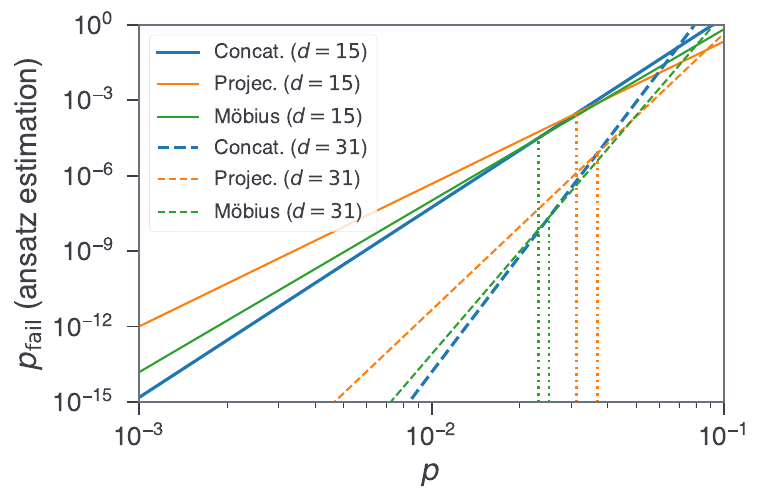}
    \caption{
        \textbf{Comparison of three matching-based decoders under bit-flip noise.}
        Logical failure rates $p_\mr{fail}$ estimated from the ansatz of Eq.~\eqref{eq:pfail_ansatz} are plotted over $p$ for two code distances $d \in \qty{15, 31}$ and three decoders: the concatenated (Concat.), projection (Proj.) \cite{delfosse2014decoding,beverland2021cost}, and Möbius MWPM decoders \cite{sahay2022decoder}.
        The parameter values used for the estimation are those in Eq.~\eqref{eq:ansatz_parameters_bitflip} for the concatenated MWPM decoder, $(\ScalingThrs{bitflip}, \alpha, \beta, \eta, d_0) = (0.087, 0.1, 1/3, 2/3, 0)$ for the projection decoder (estimated optimistically by assuming $\ScalingThrs{bitflip} = \CrossThrs{bitflip}$), and $(\ScalingThrs{bitflip}, \alpha, \beta, \eta, d_0) = (0.0801, 0.148, 0.422, 0.488, 0)$ for the Möbius decoder (reported in Ref.~\cite{sahay2022decoder}).
        The intersections of the curves of the concatenated MWPM decoder and the other two decoders are $3.12\%$ ($d=15$) and $3.69\%$ ($d=31$) for the projection decoder and $2.31\%$ ($d=15$) and $2.52\%$ ($d=31$) for the Möbius decoder.
    }
    \label{fig:decoder_comparison_bitflip}
\end{figure}

Although the obtained cross threshold $\CrossThrs{bitflip} = 8.2\%$ is lower than 8.7\% of the projection decoder \cite{delfosse2014decoding} and 9.0\% of Möbius decoder \cite{sahay2022decoder}, the concatenated MWPM decoder has a significant advantage in terms of the scaling of the failure rate over $d$ and $p$.
Namely, $\beta \approx 1/2$ for our decoder within our simulation range of $d \leq 31$, while $\beta \approx 1/3$ for the projection decoder \cite{beverland2021cost} and $\beta \approx 3/7$ for the Möbius decoder \cite{sahay2022decoder}.
The impact of this improvement is numerically presented in Fig.~\ref{fig:decoder_comparison_bitflip}, where logical failure rates $p_\mr{fail}$ are estimated using the ansatz of Eq.~\eqref{eq:pfail_ansatz} for these three decoders and two code distances $d \in \qty{15, 31}$.
The values of the parameters $(\ScalingThrs{bitflip}, \alpha, \beta, \eta, d_0)$ that we use for the estimation are $(0.087, 0.1, 1/3, 2/3, 0)$ for the projection decoder (which are optimistically guessed by assuming $\ScalingThrs{bitflip} = \CrossThrs{bitflip}$) and $(0.0801, 0.148, 0.422, 0.488, 0)$ for the Möbius decoder (which are reported in Ref.~\cite{sahay2022decoder}).
We observe that the concatenated MWPM decoder outperforms the other two decoders when $p \lessapprox 2\% \approx \CrossThrs{bitflip}/4$.

We note that the value of $\beta$ is closely related to the least-weight uncorrectable errors; namely, a decoder may be able to correct any error of weight smaller than $O(\beta d)$ for a color code with code distance $d$.
The projection decoder cannot correct specific types of errors with weights $O(d/3)$ \cite{beverland2021cost,sahay2022decoder}, which are correctable via the concatenated MWPM decoder as shown in Appendix~\ref{subsec:uncorrectable_errors_projection}.
Likewise, we can find weight-$O(3d/7)$ uncorrectable errors for the Möbius decoder \cite{sahay2022decoder}.

One may be led to guess that the concatenated MWPM decoder can correct any error of weight smaller than $O(d/2)$. 
Surprisingly, however, there exist uncorrectable errors with weights $O(3d/7)$ just as with the Möbius decoder, as illustrated in Appendix~\ref{subsec:uncorrectable_errors_concat}.
This raises the question of how the numerical analysis appears to scale as $\beta \approx 1/2$.
We speculate that small-weight uncorrectable errors with weights $< O(d/2)$ may exist only when $d$ is sufficiently large.
As evidence, weight-$O(3d/7)$ uncorrectable errors of the type described in Appendix~\ref{subsec:uncorrectable_errors_concat} exist only when $d \geq 25$.
Hence, we conjecture that $\beta$ may approach a limit of $3/7$ for sufficiently large values of $d$, although its confirmation would demand extensive computational resources.
Nonetheless, the concatenated MWPM decoder is still beneficial compared to the Möbius decoder, where $\beta$ is computed to be about $3/7$ even for small values of $d < 25$.
Note that Ref.~\cite{sahay2022decoder} also suggests a modification of the Möbius decoder that manually compares inequivalent recovery operators, which is proved to correct every error with weight less than $d/2$ when $d \leq 13$.
This improvement has not been considered in Fig.~\ref{fig:decoder_comparison_bitflip}.

An additional numerical analysis for bit-flip noise is presented in Fig.~\ref{fig:color_strategy_comparison} of Appendix~\ref{app:more_numerical_analyses}, showing that our strategy of comparing the outcomes from the sub-decoding procedures of all the three colors is indeed effective compared to performing only one or two sub-decoding procedures.

\section{Generalized concatenated MWPM decoder for circuit-level noise \label{sec:decoder_circuit_level}}

The bit-flip noise model considered in the previous section is adequate for assessing basic performance features of decoders, but it does not represent realistic noise that is relevant to fault-tolerant quantum computing.  In practice, syndrome measurements are not perfect.
We thus need to modify our decoder appropriately to accommodate circuit-level noise.
We adapt our concatenated MWPM decoder to a circuit-level noise model defined in Sec.~\ref{subsec:settings} by using the concept of \emph{detector error model} (DEM) \cite{gidney2021stim}, which is a list of independent error mechanisms.
We generalize the three steps of the concatenated MWPM decoder in Sec.~\ref{sec:decoder_bit_flip} using DEMs and employ the \emph{\stim} library \cite{gidney2021stim} to implement and analyze the decoder.

\subsection{Circuit and detectors \label{subsec:circuit}}

\begin{figure}[!t]
    \centering
    \includegraphics[width=\linewidth]{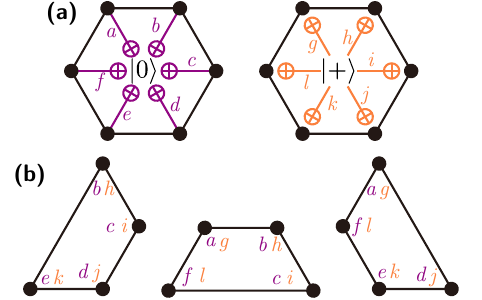}
    \caption{
        \textbf{Syndrome extraction circuit for color codes.} 
        \subfig{a} Circuits for measuring checks on hexagonal faces, which are composed of multiple \cnot gates. 
        The left (right) circuit followed by a $Z$ ($X$) measurement on the center ancillary qubit measures the $Z$-type ($X$-type) checks.
        Twelve variables $a, b, \cdots, l$ are positive integers specifying the time slices that the corresponding \cnot gates are applied.
        \subfig{b} Faces located along the boundaries and corresponding time-slice variables of the \cnot gates, colored in purple (orange) for the measurements of $Z$-type ($X$-type) checks.
        The syndrome extraction circuits are in the same form as \subfig{a}, and are therefore omitted from the figure for simplicity.
        }
    \label{fig:syndrome_extraction_circuit}
\end{figure}

We first need to clarify the circuit implementing a color code, especially its syndrome extraction.
We suppose that each face of the lattice contains two ancillary qubits (referred to as $Z$-type and $X$-type ancillary qubits), which are respectively used to extract the measurement outcomes of the $Z$-type and $X$-type checks on the face.
We consider the scenario of $T$ rounds of the logical idling gate with logical initialization and measurement, as described in Sec.~\ref{subsec:settings}.

In each round of syndrome extraction, $Z$-type ($X$-type) ancillary qubits are first initialized to $\ket{0}$ ($\ket{+}$), followed by \cnot gates between the ancillary and data qubits, concluding with the $Z$-basis ($X$-basis) measurement of the ancilla.
The arrangement of these \cnot gates are presented in Fig.~\ref{fig:syndrome_extraction_circuit}(a), where the left (right) circuit followed by a $Z$ ($X$) measurement on the center $Z$-type ($X$-type) ancillary qubit measures the $Z$-type ($X$-type) check on the face.
We have a degree of freedom on choosing the \cnot schedule (i.e., the order of applying these \cnot gates), which will be discussed later in Sec.~\ref{subsec:optimizing_cnot_schedule}.

We group the operations in the circuit in a way that each group (called a \emph{time slice}) is composed of consecutive operations applied on distinct qubits that can be performed simultaneously.
As stated in Sec.~\ref{subsec:settings}, we regard that, in each time slice, idle gates $I$ are applied on all the qubits that are not involved in any non-trivial operations.

Let $\AncMeasOutcome{f}{Z}{t}, \AncMeasOutcome{f}{X}{t} \in \qty{\pm 1}$ denote the measurement outcomes of the $Z$-type and $X$-type ancillary qubits, respectively, of a face $f$ in the $t$-th round, where $1 \leq t \leq T$.
We also define $\DataMeasOutcome{v} \in \qty{\pm 1}$ as the final measurement outcome of the data qubit located at a vertex $v$.
For every face $f$ and every pair of integers $t \in \qty{1, 2, \cdots, T+1}$ and $s \in \qty{2, 3, \cdots, T}$, the values
\begin{subequations}
\begin{align}
    \Detector{f}{Z}{t} &\coloneqq \begin{cases}
        \AncMeasOutcome{f}{Z}{1} & \text{if } t = 1, \\[0.5em]
        \AncMeasOutcome{f}{Z}{t-1} \AncMeasOutcome{f}{Z}{t} & \text{if } 2 \leq t \leq T, \\[0.5em]
        \AncMeasOutcome{f}{Z}{T} \prod_{v \in f} \DataMeasOutcome{v} & \text{if } t = T + 1,
    \end{cases} \label{eq:Z_type_detector}\\
    \Detector{f}{X}{s} &\coloneqq \AncMeasOutcome{f}{X}{s-1} \AncMeasOutcome{f}{X}{s}, \label{eq:X_type_detector}
\end{align}
\label{eq:detectors}
\end{subequations}
must be $+1$ when there are no errors.
We refer to these values as \emph{detectors} and classify them by their Pauli types and the colors of the faces; for example, $D_{f,Z}^{(t)}$ is a red $Z$-type detector if $f$ is a red face.
If a detector is $-1$, we say it is violated.

Lastly, the final measurement of the logical-$Z$ operator always has an outcome of $+1$ when there are no errors.
Thus, if $V_\mr{bdry}^{(\rbs)}$ is defined by the set of the vertices placed along the red boundary, the value
\begin{align}
    L_Z = \prod_{v \in V_\mr{bdry}^{(\rbs)}} \DataMeasOutcome{v} \label{eq:logical_observable}
\end{align}
must be $+1$ when there are no errors.
We refer to $L_Z$ as the \emph{logical observable} of our scenario.
After decoding, a correction $L_Z^{(\mr{corr})}$ is obtained and the decoding succeeds if and only if $L_Z = L_Z^{(\mr{corr})}$.
We describe $L_Z$ as `$Z$-type' because it involves only $Z$-basis measurements and shares common factors ($\qty{\DataMeasOutcome{v}}$) exclusively with $Z$-type detectors; see Eqs.~\eqref{eq:detectors} and~\eqref{eq:logical_observable}.
Note that, for an alternative scenario where the logical qubit is initialized and measured in the $\lgx$ basis, we have one $X$-type logical observable instead.

\subsection{Generalization of the concatenated MWPM decoder}

A \emph{detector error model} (DEM) is defined by a set of independent error mechanisms.
Each error mechanism, which is formally described as a 3-tuple $(q, \mathcal{D}, \mathcal{O})$, specifies the probability ($q$) that the error occurs and the set of detectors ($\mathcal{D}$) and logical observables ($\mathcal{O}$) flipped by the error.
The elements of $\mathcal{D} \cup \mathcal{O}$ are referred to as the \emph{targets} of the error mechanism.
We say that an error mechanism $(q, \mathcal{D}, \mathcal{O})$ is \emph{edge-like} if and only if $\abs{\mathcal{D}} \leq 2$.

The \stim library \cite{gidney2021stim} can be used to extract a DEM from a noisy Clifford circuit (which is composed of Clifford gates, Pauli initializations/measurements, and Pauli noise channels) provided that detectors and logical observables are annotated appropriately.
It works as follows:
If the circuit has a single-Pauli error channel $\paulichannel{q}{P}$, which applies a specific Pauli product operator $P$ with probability $q$, we can get its effect by commuting $P$ to the end of the circuit and checking the detectors and logical observables flipped by it.
If the circuit has a depolarizing channel $\depchannelone{p}$ or $\depchanneltwo{p}$, we can convert it into a sequence of single-Pauli error channels as 
\begin{align*}
    \depchannelone{p} &= \paulichannel{q_1}{X} \circ \paulichannel{q_1}{Y} \circ \paulichannel{q_1}{Z}, \\ 
    \depchanneltwo{p} &= \paulichannel{q_2}{X \otimes I} \circ \paulichannel{q_2}{X \otimes X} \circ \paulichannel{q_2}{X \otimes Y} \circ \cdots \circ \paulichannel{q_2}{Z \otimes Z},
\end{align*}
where $q_1 \coloneqq (1 - \sqrt{1 - 4p/3})/2$ and $q_2 \coloneqq [1 - (1 - 16p/15)^{1/8}]/2$.
We can then take account of these single-Pauli error channels individually.
Note that such exact decomposition might be not possible for general Pauli noise channels, but it is not important in our discussions.

In the 2D variant of the decoder presented in Sec.~\ref{sec:decoder_bit_flip}, we consider two sub-lattices (\cbs-restricted lattice $\LatticeRest{\cbs}$ and \cbs-only lattice $\LatticeMono{\cbs}$) for each sub-decoding procedure of color $\cbs$.
Analogously, for each color \cbs, we deform and decompose the DEM $\mathcal{M}$ obtained from a color-code circuit into the \emph{\cbs-restricted DEM} $\ResDEM{\cbs}$ and \emph{\cbs-only DEM} $\OnlyDEM{\cbs}$.

\begin{algorithm*}[!t]
\caption{Decomposition of a color-code detector error model (DEM)}
\label{alg:decomposing_DEM}
\KwIn{A DEM $\mathcal{M} = \qty{(q_i, \mathcal{D}_i, \mathcal{O}_i)}_i$ \newline 
A color $\cbs \in \qty{\rbs, \gbs, \bbs}$}
\KwOut{A \cbs-restricted DEM $\ResDEM{\cbs}$ \newline 
A \cbs-only DEM $\OnlyDEM{\cbs}$ \newline
A set of new `virtual' detectors $\qty{D_e}_{e \in \ResDEM{\cbs}}$
}
\tcc{Separating $Z$- and $X$-type detectors}
$\mathcal{M}_{ZX} \gets \emptyset$\;
\ForEach{$(q, \mathcal{D}, \mathcal{O}) \in \mathcal{M}$}{
    \ForEach{$P \in \qty{Z, X}$}{
        $\mathcal{D}_P \gets \qty{\text{$P$-type detectors in $\mathcal{D}$}}$\;
        $\mathcal{O}_P \gets \qty{\text{$P$-type logical observables in $\mathcal{D}$}}$\tcp*{Our scenario has only one $Z$-type logical observable, but this is to ensure generalizability.}
        \uIf{$\mathcal{D}_P \neq \emptyset$}{
            Add $(q, \mathcal{D}_P, \mathcal{O}_P)$ to $\mathcal{M}_{ZX}$\;
        }
    }
}
Compress $\mathcal{M}_{ZX}$ (i.e., successively merge a pair of error mechanisms of $\mathcal{M}_{ZX}$ that have the same targets while updating the probability as $q = q_1 + q_2 - 2q_1q_2$, where $q_1$ and $q_2$ are respectively the probabilities of the two error mechanisms)\;
\tcc{Constructing $\ResDEM{\cbs}$}
$\ResDEM{\cbs} \gets \emptyset$\;
\ForEach{$(q, \mathcal{D}, \mathcal{O}) \in \mathcal{M}_{ZX}$}{
    $\mathcal{D}_\mr{\neg\cbs} \gets \qty{\text{Detectors in $\mathcal{D}$ that are not \cbs-colored}}$\;
    \uIf{$0 < \abs{\mathcal{D}_\mr{\neg\cbs}} \leq 2$}{
        Add $(q, \mathcal{D}_\mr{\neg\cbs}, \emptyset)$ to $\ResDEM{\cbs}$\;
    }
}
Compress $\ResDEM{\cbs}$\;
$\qty{D_e}_{e \in \ResDEM{\cbs}} \gets \text{A set of new detectors}$\;
\tcc{Constructing $\OnlyDEM{\cbs}$}
$\OnlyDEM{\cbs} \gets \emptyset$\;
\ForEach{$(q, \mathcal{D}, \mathcal{O}) \in \mathcal{M}_{ZX}$}{
    $\mathcal{D}_\mr{\neg\cbs} \gets \qty{\text{Detectors in $\mathcal{D}$ that are not \cbs-colored}}$\;
    \uIf{$\mathcal{D}_\mr{\neg\cbs} = \emptyset$ and $\abs{\mathcal{D}} \leq 2$}{
        Add $(q, \mathcal{D}, \mathcal{O})$ to $\OnlyDEM{\cbs}$\;
    }\ElseIf{$0 < \abs{\mathcal{D}_\mr{\neg\cbs}} \leq 2$ and $\abs{\mathcal{D} \setminus \mathcal{D}_\mr{\neg\cbs}} \leq 1$}{
        $e \gets \text{An element of $\ResDEM{\cbs}$ with detectors $\mathcal{D}_\mr{\neg\cbs}$}$ (which uniquely exists)\;
        Add $(q, \mathcal{D} \setminus \mathcal{D}_\mr{\neg\cbs} \cup \qty{D_e}, \mathcal{O})$ to $\OnlyDEM{\cbs}$\;
    }
}
\end{algorithm*}

The step-by-step instructions to construct $\ResDEM{\cbs}$ and $\OnlyDEM{\cbs}$ from $\mathcal{M}$ are presented in Algorithm~\ref{alg:decomposing_DEM}.
These can be outlined for $\cbs=\rbs$ as follows:
First (in lines~1--10), we define $\mathcal{M}_{ZX}$ from $\mathcal{M}$ by separating each error mechanism that affects both $Z$- and $X$-type detectors into two independent mechanisms, ensuring every error mechanism affects only one detector type.
If the logical observable $L_Z$ is a target of an error mechanism to be separated, $L_Z$ is included only in the $Z$-type part after the separation.
Note that this process is equivalent to ignoring correlations between $X$ and $Z$ errors originated from $Y$ errors or two-qubit errors such as $X \otimes Z$.
We then compress $\mathcal{M}_{ZX}$; namely, we merge error mechanisms of $\mathcal{M}_{ZX}$ that have the same set of targets while updating the probabilities properly.
After that (in lines 11--18), $\ResDEM{\rbs}$ is constructed by removing red detectors and logical observables from each error mechanism of $\mathcal{M}_{ZX}$ and compressing it.
Importantly, we introduce a new \emph{virtual detector} $D_e$ for each error mechanism $e$ in $\ResDEM{\rbs}$.
Lastly (in lines 19--28), we build $\OnlyDEM{\rbs}$ by replacing the green and blue detectors of each error mechanism of $\mathcal{M}_{ZX}$ with the corresponding virtual detector.
As a consequence, only red detectors, virtual detectors, and logical observables are involved in $\OnlyDEM{\rbs}$.

In addition to the above description, Algorithm~\ref{alg:decomposing_DEM} also contains processes to leave only edge-like mechanisms in the decomposed DEMs (see the conditions in lines~14, 22, and 24).
By doing so, each of the DEMs can be expressed as a weighted graph $G$ whose vertices comprise detectors and an additional boundary vertex $v_\mr{bdry}$.
Namely, each error mechanism $(q, \mathcal{D}, \mathcal{O})$ corresponds to an edge of $G$ with weight $\log \qty[(1-q)/q]$, which connects the two detectors in $\mathcal{D}$ (if $\abs{\mathcal{D}} = 2$) or the only detector in $\mathcal{D}$ and $v_\mr{bdry}$ (if $\abs{\mathcal{D}} = 1$).
This graph can be used to perform the MWPM algorithm that predicts one of the most probable combinations of error mechanisms consistent with given violated detectors.
If a logical observable is flipped by an odd number of these error mechanisms, its correction is $-1$; otherwise, it is $+1$.

With the above ingredients, we can finally describe how the concatenated MWPM decoder works.
Denoting the set of violated \cbs-colored detectors as $\sigma^{(\cbs)}$ for each $\cbs \in \qty{\rbs, \gbs, \bbs}$, the decoder works as follows:
\begin{enumerate}
    \item Obtain the red-restricted DEM $\ResDEM{\rbs}$, red-only DEM $\OnlyDEM{\rbs}$, and set of virtual detectors $\qty{D_e}_{e \in \ResDEM{\rbs}}$ from the DEM of the original circuit through Algorithm~\ref{alg:decomposing_DEM}.
    \item Perform the MWPM algorithm on $\ResDEM{\rbs}$ with the input $\sigma^{(\gbs)} \cup \sigma^{(\bbs)}$ and obtain a least-weight error set $E_\rbs \subseteq \ResDEM{\rbs}$.
    \item Perform the MWPM algorithm on $\OnlyDEM{\rbs}$ with the input $\sigma^{(\rbs)} \cup \qty{D_e \mid e \in E_\rbs}$ and obtain a correction $L_Z^{(\mr{corr},\rbs)} \in \qty{\pm 1}$ and the corresponding total weight $w_\rbs$ of predicted errors.
    \item Repeat the above steps for the other two colors and obtain $L_Z^{(\mr{corr},\gbs)}$ and $L_Z^{(\mr{corr},\bbs)}$ with the corresponding total weights $w_\gbs$ and $w_\bbs$. 
    \item Return $L_Z^{(\mr{corr},\cbs)}$ where $\cbs \coloneqq \argmin_{\cbs' \in \qty{\rbs, \gbs, \bbs}} w_{\cbs'}$.
\end{enumerate}

\subsection{Optimization of the CNOT schedule \label{subsec:optimizing_cnot_schedule}}

\begin{figure}[!t]
    \centering
    \includegraphics[width=\linewidth]{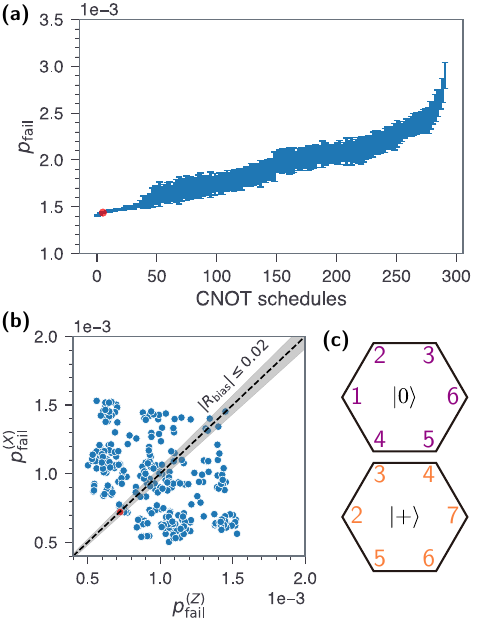}
    \caption{
        \textbf{Performance comparison of 292 length-7 CNOT schedules.}
        \subfig{a} The failure rate $p_\mr{fail}$ of $T=7$ rounds of the logical idling gate when using the concatenated MWPM decoder is evaluated for code distance $d=7$ under the circuit-level noise model of strength $p = 10^{-3}$ for each of these \cnot schedules (sorted by the values of $p_\mr{fail}$ in ascending order).
        The error bars correspond to the 99\% CIs.
        The selected optimal schedule is $[2, 3, 6, 5, 4, 1; 3, 4, 7, 6, 5, 2]$, which is marked as a red dot.
        \subfig{b} The $Z$- and $X$-failure rates $p_\mr{fail}^{(Z)}$, $p_\mr{fail}^{(X)}$ for these \cnot schedules are visualized as a scatter plot.
        The dashed line plots $p_\mr{fail}^{(X)} = p_\mr{fail}^{(Z)}$ and the gray area indicates the region of $\abs{R_\mathrm{bias}} \leq 0.02$.
        The selected optimal schedule is marked as a red dot.
        \subfig{c} Selected optimal schedule following the same graphical representation as in Fig.~\ref{fig:syndrome_extraction_circuit}(a).
    }
    \label{fig:cnot_schedules_comparison}
\end{figure}

The time order of \cnot gates for syndrome extraction, called the \cnot schedule, needs to be optimized carefully before analyzing the performance of the decoder.
We use a similar method as Ref.~\cite{beverland2021cost} to determine it.
A brief review of this is as follows:
The \cnot schedule can be specified by a tuple of twelve positive integers $\mathcal{A} = \qty[a, b, c, d, e, f; g, h, i, j, k, l]$, which contains all the integers from 1 to $\max\mathcal{A}$ (called the length of the schedule).
Each integer indicates the time slice at which the corresponding \cnot gate presented in Fig.~\ref{fig:syndrome_extraction_circuit} is applied; that is, we first apply the \cnot gates with integer 1, then apply those with integer 2, and so on.
Several conditions need to be imposed on $\mathcal{A}$ since each qubit can be involved in at most one \cnot gate per time slice and $Z$-type and $X$-type syndrome measurements should not interfere each other; see Sec.~II~C of Ref.~\cite{beverland2021cost} for their explicit descriptions.
No \cnot schedules with length less than 7 satisfy these conditions.
However, there are 876 valid length-7 schedules and we can leave only 292 among them by removing every schedule equivalent to another schedule up to a symmetry.\footnote{This number (292) of valid length-7 schedules is inconsistent with the number (234) reported in Ref.~\cite{beverland2021cost}. 
We discussed it with the first author of Ref.~\cite{beverland2021cost}, but could not identify the precise reason of this discrepancy.
We have verified using \stim that all of the 292 schedules that we identify give the same detectors as intended.  We also note that, should we have mistakenly included additional redundant schedules, this would not affect finding an optimal schedule.}

We note that our simulating scenario described in Sec.~\ref{subsec:settings} is only for computing the $\ov{Z}$-failure rate $p_\mr{fail}^{(Z)}$ (i.e., failure rate of the $Z$ observable).
Nonetheless, we can utilize the fact that the $\ov{X}$-failure rate $p_\mr{fail}^{(X)}$ is equal to the $\ov{Z}$-failure rate when using the \cnot schedule with the ``$Z$-part'' (i.e., first six integers) and ``$X$-part'' (i.e., last six integers) reversed from the original one.
For example, the $\ov{X}$-failure rate for the schedule $[2, 3, 6, 5, 4, 1; 3, 4, 7, 6, 5, 2]$ can be computed from our scenario by instead using the schedule $[3, 4, 7, 6, 5, 2; 2, 3, 6, 5, 4, 1]$.
For a given \cnot schedule, we compute the logical failure rate as $p_\mr{fail} = p_\mr{fail}^{(X)} + p_\mr{fail}^{(Z)}$ and quantify the bias between them as
\begin{align}
    R_\mr{bias} \coloneqq \log_{10} \frac{p_\mathrm{fail}^{(Z)}}{p_\mathrm{fail}^{(X)}},
    \label{eq:bias_coeff}
\end{align}
which is zero when there is no bias.

\begin{figure*}[!t]
    \centering
    \includegraphics[width=\textwidth]{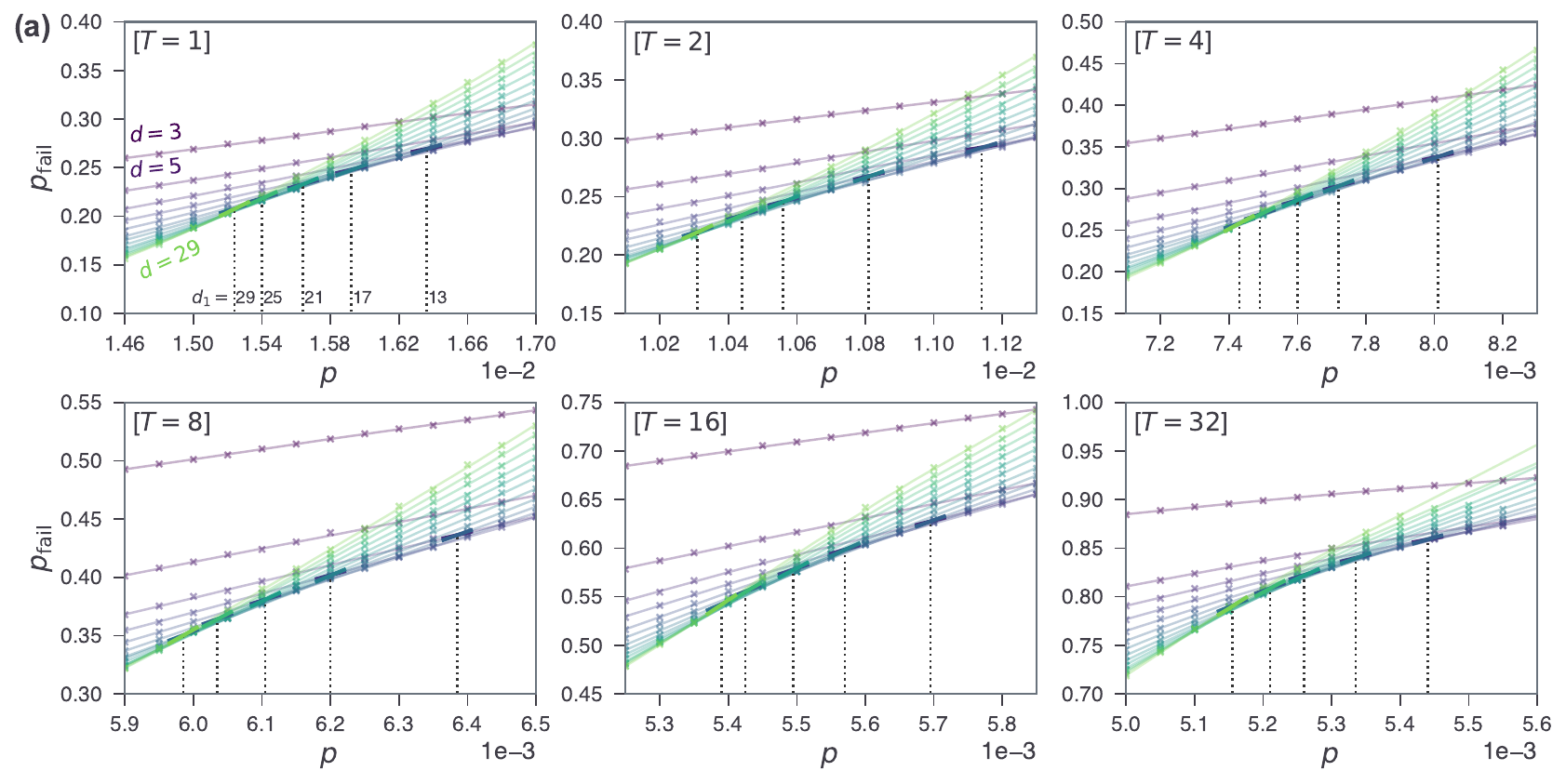}
    \includegraphics[width=0.86666667\textwidth]{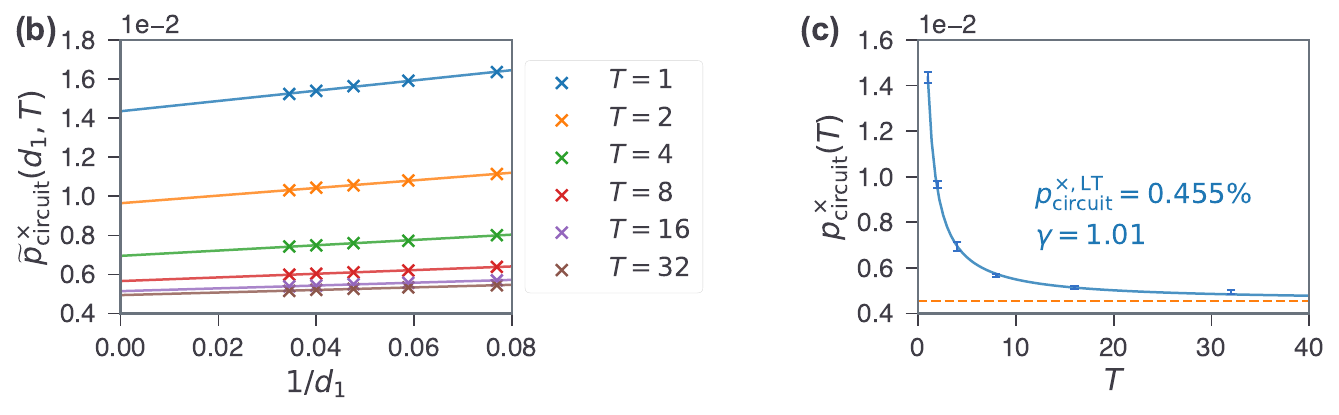}
    \caption{
        \textbf{Numerical analysis of the decoder under near-threshold circuit-level noise.}
        We consider $T \in \qty{1, 2, 4, 8, 16, 32}$ rounds of the logical idling gate of the triangular color code with code distance $d \in \qty{3, 5, 7, \cdots, 29}$ under the circuit-level noise model of strength $p$.
        \subfig{a} Logical failure rates $p_\mr{fail} = p_\mr{fail}^{(X)} + p_\mr{fail}^{(Z)}$ are plotted over $p$ for each $T$ and $d$.
        Each data point (marked as `X') has the 99\% CI of $\pm 10^{-3}$.
        For each $d$, the LOWESS regression of $p_\mr{fail}$ with fraction $2/3$ is drawn as a solid line.
        For each $d_1 \in \qty{13, 17, 21, 25, 29}$, the crossing of the two regression lines of $d=d_1$ and $d=\qty(d_1 + 1)/2$ is highlighted in dark color.
        \subfig{b} The values of $p$ at these crossings $\CrossThrsVar{circuit} (d_1, T)$ are plotted over $1/d_1$ for each $T$ with their linear regressions.
        The cross thresholds $\CrossThrs{circuit}(T)$ are estimated by the intersections of these regressions with the vertical axis.
        \subfig{c} The 99\% CIs of the estimations of $\CrossThrs{circuit}(T)$ are presented with a fitted curve in the form of Eq.~\eqref{eq:long_term_thrs_ansatz}, where $\CrossThrsLT{circuit} \approx 0.455\%$ and $\gamma \approx 1.01$.
        The orange dashed line plots $\CrossThrs{circuit}(T) = \CrossThrsLT{circuit}$.
    }
    \label{fig:circuit_numerical_analysis_near_thrs}
\end{figure*}

To find the optimal \cnot schedule, we evaluate $p_\mr{fail}$ by using the concatenated MWPM decoder when $d=T=7$ for the above 292 valid length-7 \cnot schedules under the circuit-level noise model of strength $p = 10^{-3}$.
We select $p = 10^{-3}$ since we mainly concern the sub-threshold performance of our decoder when $p \leq 10^{-3}$.
The evaluated values of $p_\mr{fail}$ are plotted in ascending order in Fig.~\ref{fig:cnot_schedules_comparison}(a).
In addition, Fig.~\ref{fig:cnot_schedules_comparison}(b) shows the distribution of $p_\mr{fail}^{(Z)}$ and $p_\mr{fail}^{(X)}$ for these schedules.
We select the optimal \cnot schedule as the one that gives the smallest $p_\mr{fail}$ among the schedules having $\abs{R_\mr{bias}} \leq 0.02 \approx \log_{10} 1.05$, which is visualized as a gray area in Fig.~\ref{fig:cnot_schedules_comparison}(b).
The selected \cnot schedule is
\begin{align}
    \mathcal{A} = [2, 3, 6, 5, 4, 1; 3, 4, 7, 6, 5, 2],
    \label{eq:optimal_cnot_schedule}
\end{align}
which is marked as red dots in Figs.~\ref{fig:cnot_schedules_comparison}(a) and~(b) and visualized in Fig.~\ref{fig:cnot_schedules_comparison}(c).
The corresponding failure rates and bias (99\% CI) are
\begin{align*}
    p_\mr{fail} ={}& (1.437 \pm 0.005) \times 10^{-3}, \\
    p_\mr{fail}^{(Z)} = p_\mr{fail}^{(X)} ={}& (7.19 \pm 0.04) \times 10^{-4}, \\
    R_\mr{bias} ={}& 0.000 \pm 0.003.
\end{align*}
Note that the failure rate for the worst-performing \cnot schedule is $p_\mr{fail} = (2.9 \pm 0.1) \times 10^{-3}$, which is approximately twice that of the best case.
In Appendix~\ref{app:bias_origin}, we discuss the origin of the bias and explain why the degree of this bias varies with the \cnot schedule.

Although we here choose the schedule to have a sufficiently small bias, biased schedules may be intentionally used depending on the task.
For instance, in the 15-to-1 magic state distillation scheme \cite{bravyi2005universal}, $\ov{X}$ errors on ancillary logical qubits are more detrimental than $\ov{Z}$ errors, as the latter can be detected by the final $\ov{X}$ measurements \cite{litinski2019magic}.
Therefore, using \cnot schedules biased towards $\ov{Z}$ errors (i.e., $R_\mr{bias} > 0$) for the ancillary logical qubits can be advantageous.


\subsection{Performance analysis}

We now analyze the performance of the decoder under circuit-level noise models when the \cnot schedule is set to the optimal one in Eq.~\eqref{eq:optimal_cnot_schedule}.
As in Sec.~\ref{subsec:bitflip_performance_analysis}, we first examine the near-threshold region to determine the cross threshold and then investigate the sub-threshold scaling of the logical failure rate.

For near-threshold simulations, we consider $T \in \qty{1, 2, 4, 8, 16, 32}$ rounds of the logical idling gate with code distance $d \in \qty{3, 5, 7, \cdots, 29}$.
The evaluated logical failure rates $p_\mr{fail}$ are plotted over the noise strength $p$ in Fig.~\ref{fig:circuit_numerical_analysis_near_thrs}(a) for each $T$ and $d$.
Unlike the case of bit-flip noise plotted in Fig.~\ref{fig:bitflip_numerical_analysis}(a), the regression lines for each value of $T$ do not clearly intersect at a single point.
Therefore, to determine the cross threshold $\CrossThrs{circuit}(T)$, we evaluate $\CrossThrsVar{circuit} (d_1, T)$, which is the crossing of the two lines of $d = d_1$ and $d = (d_1 + 1)/2$ for $d_1 \equiv 1 \pmod{4}$, while varying $d_1$ and fit the data into the ansatz
\begin{align*}
    \CrossThrsVar{circuit} (d_1, T) = \frac{A}{d_1} + \CrossThrs{circuit}(T),
\end{align*}
where $A$ is a free parameter.
The values of $\CrossThrsVar{circuit} (d_1, T)$ for $d_1 \in \qty{13, 17, 21, 25, 29}$ are highlighted in dark color in Fig.~\ref{fig:circuit_numerical_analysis_near_thrs}(a) and plotted in Fig.~\ref{fig:circuit_numerical_analysis_near_thrs}(b) against $1/d_1$ with their linear regressions.
The obtained cross thresholds $\CrossThrs{circuit}(T)$ are visualized in Fig.~\ref{fig:circuit_numerical_analysis_near_thrs}(c).
By fitting them into Eq.~\eqref{eq:long_term_thrs_ansatz}, we get the long-term cross threshold $\CrossThrsLT{circuit} \approx 0.455\%$ with $\gamma \approx 1.01$.

\begin{figure}[!t]
	\centering
	\includegraphics[width=\linewidth]{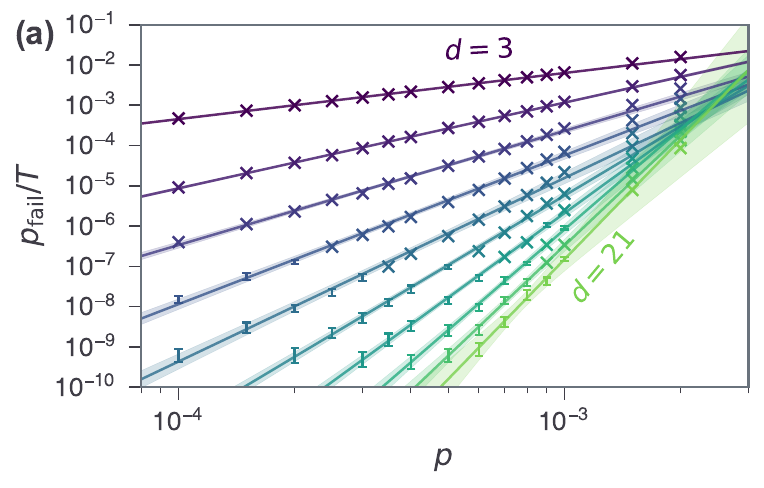}
    \includegraphics[width=\linewidth]{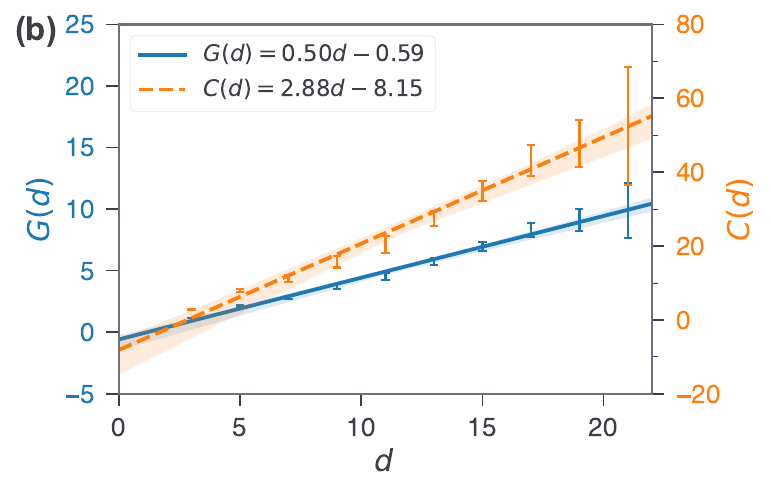}
	\caption{
        \textbf{Numerical analysis of the decoder under sub-threshold circuit-level noise.}
        We consider $T = 4d$ rounds of the logical idling gate of the triangular color code with code distance $d \in \qty{3, 5, \cdots, 21}$ under the circuit-level noise model of strength $p$.
        \subfig{a} Logical failure rates per round ($p_\mr{fail}/T$) are plotted over $p$ for various code distances.
        Each data point is marked as an error bar indicating its 99\% CI or as an `X' symbol if its relative margin of error (i.e., the ratio of half the width of its CI to its center) is less than 10\%.
        For each $d$, the data within the range of $p \leq 10^{-3}$ are used for the regression in the form of Eq.~\eqref{eq:log_pfail_ansatz}, which is drawn as a solid line.
        The 99\% CIs of the regression estimates are depicted as shaded regions around the lines.
        \subfig{b} The 99\% CIs of $G(d)$ and $C(d)$ in Eq.~\eqref{eq:log_pfail_ansatz} are plotted over $d$.
        Their linear regressions are shown as solid/dashed lines with shaded regions indicating the 99\% CIs of the estimates.
    }
	\label{fig:pfails_low_errors_circuit}
\end{figure}

We next analyze the sub-threshold scaling of logical failure rates.
We consider $T = 4d$ rounds of the logical idling gate with code distance $d \in \qty{3, 5, 7, \cdots, 21}$ and compute its logical failure rate per round ($p_\mr{fail}/T$) as presented in Fig.~\ref{fig:pfails_low_errors_circuit}(a).
For each $d$, the data with $p \leq 10^{-3}$ are fitted into the ansatz of Eq.~\eqref{eq:log_pfail_ansatz}, whose gradient $G(d)$ and constant $C(d)$ are plotted in Fig.~\ref{fig:pfails_low_errors_circuit}(b).
We obtain the regressions of $G(d)$ and $C(d)$ against $d$ as
\begin{align*}
    G(d) ={}& (0.50 \pm 0.05) (d - 12) + (5.4 \pm 0.3), \\
    C(d) ={}& (2.9 \pm 0.4) (d - 12) + (26 \pm 2),
\end{align*}
by selecting $d_0 = 12$, where the error terms are the 99\% CIs of the estimations.
The corresponding ansatz parameters in Eq.~\eqref{eq:pfail_ansatz} are
\begin{align*}
    \begin{split}
        &\ScalingThrs{circuit} \approx 3.2 \times 10^{-3}, \quad \alpha \approx 8.4 \times 10^{-3}, \\ 
        &\beta \approx 0.50, \quad \eta \approx 5.4, \quad d_0 = 12.
    \end{split}
\end{align*}

\begin{figure*}[!t]
	\centering
	\includegraphics[width=\textwidth]{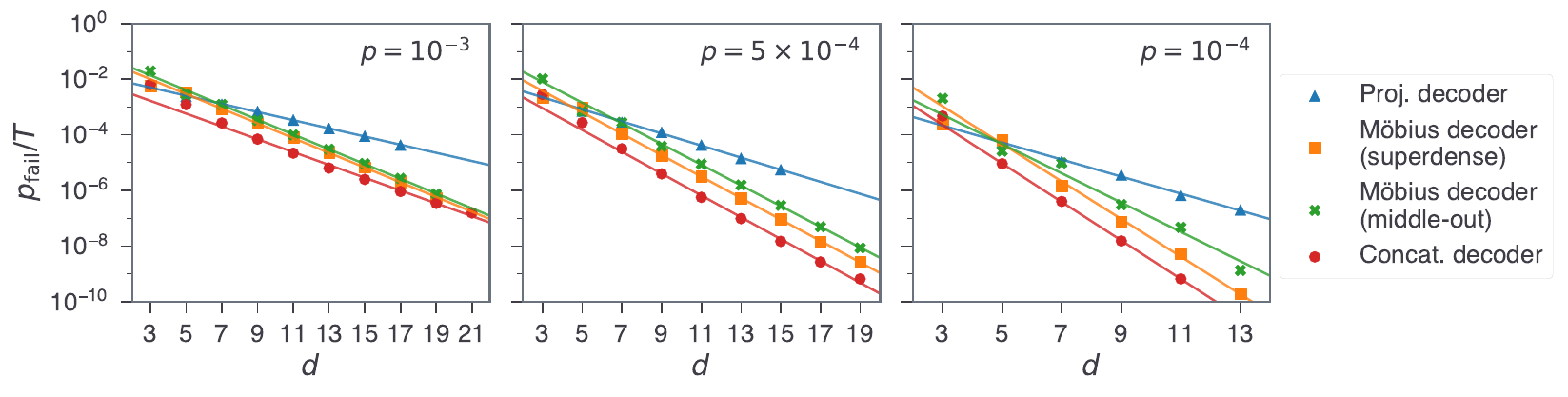}
	\caption{
        \textbf{Comparison of three matching-based decoders under circuit-level noise.}
        Logical failure rates per round $p_\mr{fail}/T$ are estimated by using three different decoders (projection, Möbius, and concatenated MWPM decoders) across different code distances $d$ at varying noise strengths $p \in \qty{10^{-3}, 5 \times 10^{-4}, 10^{-4}}$.
        The data for the projection decoder are from Fig.~17 of Ref.~\cite{beverland2021cost}.
        The data for the Möbius decoder are from Figs.~11 and~12 of Ref.~\cite{gidney2023new}, which respectively correspond to the superdense and middle-out color code circuits (two types of syndrome extraction circuits proposed in Ref.~\cite{gidney2023new}).
        The solid lines represent the linear regressions of $\log (p_\mr{fail}/T)$ against $d$.
    }
	\label{fig:decoder_comparison_circuit}
\end{figure*}

We lastly compare the performance of the concatenated MWPM decoder with that of previous decoders: the projection and Möbius decoders.
In Fig.~\ref{fig:decoder_comparison_circuit}, we present the logical failure rates per round $p_\mr{fail}/T$ estimated by using these decoders across different code distances $d$ at three noise strengths $p \in \qty{10^{-3}, 5 \times 10^{-4}, 10^{-4}}$.
The data for the projection and Möbius decoders are originated from Refs.~\cite{beverland2021cost} and~\cite{gidney2023new}, respectively.
The figure shows that the projection decoder significantly underperforms the other two due to its suboptimal scaling against $d$.
The scaling against $d$ is comparably similar for the other two decoders (when $p \lessapprox 5 \times 10^{-4}$); however, the concatenated MWPM decoder achieves logical failure rates that are approximately 3--7 times lower than the Möbius decoder.

In Figs.~\ref{fig:near_threshold_paulis}--\ref{fig:bias_scatter} of Appendix~\ref{app:more_numerical_analyses}, we present separate analyses of the $\ov{Z}$- and $\ov{X}$-failure rates of the decoder, together with analysis of the bias $R_\mr{bias}$ defined in Eq.~\eqref{eq:bias_coeff}, while setting the \cnot schedule to Eq.~\eqref{eq:optimal_cnot_schedule} as well.
The results are summarized as follows:
The long-term cross thresholds are estimated as $0.460\%$ for the $\ov{Z}$-failure and $0.452\%$ for the $\ov{X}$-failure, which differ by about 1.8\% (0.008\%p).
However, this difference likely falls within the range of statistical error, as the 99\% CIs for the estimations of $R_\mr{bias}$ almost always include zero, implying that the null hypothesis of $R_\mr{bias} = 0$ cannot be rejected at the 99\% significance level for these cases.

\section{Remarks \label{sec:remarks}}

In this work, we introduced the concatenated minimum-weight perfect matching (MWPM) decoder processed by the concatenation of two rounds of MWPM per color, which is applicable not only to simple bit-flip noise but also to realistic circuit-level noise.
The decoder is based on the idea that the outcome obtained from decoding on a restricted lattice can serve as additional `virtual syndrome data', which undergo a subsequent decoding round in conjunction with remaining syndrome data.

We numerically analyzed the performance of the decoder in various aspects: We considered both bit-flip and circuit-level noise models and investigated near-threshold and sub-threshold behaviors of the logical failure rate $p_\mr{fail}$.
We found that the decoder has the thresholds of 8.2\% for bit-flip noise and 0.46\% for circuit-level noise, which are comparable with those of previous matching-based decoders such as the projection decoder \cite{delfosse2014decoding,chamberland2020triangular,beverland2021cost,kubica2023efficient,zhang2024facilitating} and Möbius MWPM decoder \cite{sahay2022decoder,gidney2023new}.
Remarkably, we verified that the decoder approaches a scaling of $p_\mr{fail} \sim p^{d/2}$, where $p$ is the noise strength and $d$ is the code distance, at least within our simulation range ($d \lessapprox 31$ for bit-flip noise and $d \lessapprox 21$ for circuit-level noise).
As a consequence, it outperforms previous matching-based decoders in terms of their sub-threshold failure rates across both bit-flip and circuit-level noise, as visualized in Figs.~\ref{fig:decoder_comparison_bitflip} and~\ref{fig:decoder_comparison_circuit}.
We therefore anticipate that our decoder enhances the practicality of employing color codes in quantum computing, which has been considered less viable than surface codes due to its logical failure rate performance despite its advantage in resource efficiency \cite{thomsen2024lowoverhead}.
We distributed a python module implementing the decoder on Github \cite{colorcodestim} so that other researchers can use it.

We consider several future directions related to this work.
All the analyses in this work are based on the logical idling gate, which is generally suitable for initial studies of a decoder.
However, to use it in real implementations, we should also consider nontrivial operations.
In particular, it will be worth investigating how the decoder should be modified to handle domain walls required for lattice surgery \cite{kesselring2024anyon}.
Additionally, syndrome extraction circuits may be able to be optimized further beyond the simple circuit in Fig.~\ref{fig:syndrome_extraction_circuit}.
For example, Ref.~\cite{gidney2023new} suggests two circuits: superdense and middle-out circuits.
Superdense circuits use two ancilla qubits per face, which are prepared in a Bell pair so that superdense coding can be employed.
Middle-out circuits do not have ancilla qubits, thereby significantly reducing resource overheads at the cost of a slight decrease in fault tolerance.
It will be interesting to see how the concatenated MWPM decoder performs with such circuits.

Another promising direction for future research is to apply our concatenation approach to the Möbius decoder \cite{sahay2022decoder,gidney2023new}. 
Specifically, instead of considering the three colors independently, we could run the first-round MWPM on the unified lattice, as done in the Möbius decoder. 
A key challenge lies in constructing the lattice for the second-round MWPM and lifting the first-round MWPM outcomes from the Möbius decoder to this lattice.
We have two preliminary ideas for this: We could either (i) combine the three monochromatic lattices by introducing extra edges to account for errors at the boundaries between different colors, or (ii) keep each monochromatic lattice independent and lift each boundary error to two nodes belonging to distinct monochromatic lattices. 
Developing and testing these ideas further would be a worthwhile direction for future work.

\section*{Acknowledgements}

We thank Felix Thomsen, Nicholas Fazio, Samuel C. Smith, Benjamin J. Brown, Michael E. Beverland, and Kaavya Sahay for helpful discussions and comments. This work is supported by the Australian Research Council via the Centre of Excellence in Engineered Quantum Systems (EQUS) project number CE170100009.  This article is based upon work supported by the Defense Advanced Research Projects Agency (DARPA) under Contract No.\ HR001122C0063. Any opinions, findings and conclusions or recommendations expressed in this article are those of the author(s) and do not necessarily reflect the views of the Defense Advanced Research Projects Agency (DARPA).

\bibliographystyle{quantum}
\bibliography{bibliography}

\widetext

\appendix

\section{Proof of the validity of the concatenated MWPM decoder \label{app:validity_proof}}

In this appendix, we formally prove that the 2D version of the concatenated MWPM decoder described in Sec.~\ref{sec:decoder_bit_flip} indeed works well.
Namely, we show that the decoder always identifies a proper prediction of $X$ errors consistent with the syndrome.
We additionally validate that any outcome obtained from the projection decoder is a valid matching for the concatenated MWPM decoder, which implies that the former cannot outperform the latter.
We here use the same mathematical notations as used in Sec.~\ref{sec:decoder_bit_flip}.

\subsection{Notations and definitions \label{subsec:validity_proof_notations}}

Let $\mathcal{L}$ be a trivalent three-colorable lattice. 
We assume that the lattice does not have boundaries; we will consider boundaries later.
We consider the lattices: the dual lattice $\mathcal{L}^*$, the decoding hypergraph $\mathcal{H}$, the red/green/blue-restricted lattice $\LatticeRest{\rbs/\neg \gbs/\neg \bbs}$, the red/green/blue-only lattice $\LatticeMono{\rbs/\gbs/\bbs}$.

Let $\mathcal{L}^*$ be the dual lattice of $\mathcal{L}$, whose vertices are three-colorable so that endpoints of an edge are distinctly colored and whose faces are of the form $\qty{\qty{v_\rbs,v_\gbs},\qty{v_\gbs,v_\bbs},\qty{v_\bbs,v_\rbs}}$, where $v_\rbs$, $v_\gbs$, and $v_\bbs$ are respectively red, green, and blue vertices. 
The color of an edge is chosen to be distinct from the colors of its endpoints. 
We associate with the dual lattice its cellular homology complex $(C_0(\mathcal{L}^*),C_1(\mathcal{L}^*),C_2(\mathcal{L}^*))$, where $C_0(\mathcal{L}^*)$ is an $\mathbbm{F}_2$ (binary) vector space with the basis set $\Delta_0(\mathcal{L}^*)$, and accordingly for $C_1(\mathcal{L}^*)$, $\Delta_1(\mathcal{L}^*)$ and $C_2(\mathcal{L}^*)$, $\Delta_2(\mathcal{L}^*)$.
We equip the complex with linear boundary maps $\partial_1^*: C_1(\mathcal{L}^*) \rightarrow C_0(\mathcal{L}^*)$ and $\partial_2^*: C_2(\mathcal{L}^*) \rightarrow C_1(\mathcal{L}^*)$ such that $\partial_1^*(\qty{u,v}) = u+v$ and $\partial_2^*(\qty{e_1,e_2,...,e_m}) = e_1+e_2+...+e_m$.
For each color $\cbs \in \qty{\rbs, \gbs, \bbs}$, we denote $\latelmcolor{\mathcal{L}^*}{0}{\cbs} = \qty{v \in \latelm{\mathcal{L}^*}{0} \;\vert\; \text{$v$ is \cbs-colored}}$ and accordingly for $\latelmcolor{\mathcal{L}^*}{1}{\cbs}$.

The decoding hypergraph $\mathcal{H}$ has vertices $\mathcal{V} = \Delta_0 (\mathcal{L}^*)$, hyper-edges $\varepsilon = \qty{ \cup_{e \in f} e \; \vert \; f \in \Delta_2 (\mathcal{L}^*) }$ (such that each hyper-edge is of the form $\qty{v_\rbs,v_\gbs,v_\bbs}$), and hyper-faces $\mathcal{F} = \qty{ \qty{e \in \Delta_1 (\mathcal{L}^*) \;\vert\; v \in e} \;\vert\; v \in \Delta_0(\mathcal{L}^*)}$. 
On this lattice, we define a 2-complex with an $\mathbbm{F}_2$ vector space $C_0(\mathcal{H})$ spanned by the $\mathcal{V}$, $C_1(\mathcal{H})$ spanned by $\varepsilon$, and $C_2(\mathcal{H})$ spanned by $\mathcal{F}$. The boundary maps on this complex are defined as
\begin{align*}
    \partial_1^\mathcal{H} \qty(\qty{v_\rbs,v_\gbs,v_\bbs}) &= v_\rbs+v_\gbs+v_\bbs, \\
    \partial_2^\mathcal{H} \qty(\qty{e_1,e_2, \cdots, e_m}) &= e_1+e_2 + \cdots + e_m.
\end{align*}

The red-restricted lattice $\LatticeRest{\rbs}$ is the subgraph of $\mathcal{L}^*$ over vertices $\Delta_0(\LatticeRest{\rbs}) = \latelmcolor{\mathcal{L}^*}{0}{\gbs} \cup \latelmcolor{\mathcal{L}^*}{0}{\bbs}$. 
We associate it with its cellular homology complex $C_0(\LatticeRest{\rbs})$. 
We define a linear projection operator $\pi_0^{(\rbs)}: C_0(\mathcal{L}^*) \rightarrow C_0(\mathcal{L}^*)$ that satisfies
\begin{align*}
\pi_0^{(\rbs)}(v) = \begin{cases} 
      v & \text{if } v \notin \latelmcolor{\mathcal{L}^*}{0}{\rbs}, \\
      0 & \text{otherwise}
   \end{cases}
\end{align*}
for each $v \in \latelm{\mathcal{L}^*}{0}$ such that $\Im(\pi_0^{(\rbs)}) = C_0(\LatticeRest{\rbs})$.

The red-only lattice $\LatticeMono{\rbs}$ is defined as
\begin{align*}
    \Delta_0\qty(\LatticeMono{\rbs}) &= \latelmcolor{\mathcal{L}^*}{0}{\rbs} \cup \latelmcolor{\mathcal{L}^*}{1}{\rbs}, \\
    \Delta_1\qty(\LatticeMono{\rbs}) &= \qty{\qty{v_\rbs,\qty{v_\gbs,v_\bbs}} \subset \Delta_0\qty(\LatticeMono{\rbs}) \;\vert\;\ \qty{v_\rbs,v_\gbs,v_\bbs} \in \varepsilon}, \\
    \Delta_2\qty(\LatticeMono{\rbs}) &= \left\{ \qty{ \qty{v_\rbs,\qty{v_\gbs,v_\bbs}} \in \Delta_1\qty(\LatticeMono{\rbs}) \;\vert\; v \in \qty{v_\rbs,v_\gbs,v_\bbs}} \;\middle\vert\; v \in \latelmcolor{\mathcal{L}^*}{0}{\gbs} \cup \latelmcolor{\mathcal{L}^*}{0}{\bbs}\right\}, \\
\end{align*}
which respectively span $C_0\qty(\LatticeMono{\rbs})$, $C_1\qty(\LatticeMono{\rbs})$, and $C_2\qty(\LatticeMono{\rbs})$ that form its cellular homology complex.
Note that $C_1(\LatticeMono{\rbs})$ and $C_1(\mathcal{H})$ are isomorphic under the mapping of bases $\qty{v_\rbs,\qty{v_\gbs,v_\bbs}} \leftrightarrow \qty{v_\rbs,v_\gbs,v_\bbs}$.
Therefore, we will treat edges in the red-only lattice as hyper-edges when considering complexes.
We also note that $C_0(\LatticeMono{\rbs}) = C_1(\LatticeRest{\rbs}) \;\oplus\; \Im(\identity - \pi^{(\rbs)}_0)$.
We define $P_{C_1(\LatticeRest{\rbs})}$ and $P_{\Im(\identity - \pi_0^{(\rbs)})}$ as the projectors from $C_0(\LatticeMono{\rbs})$ onto these two subspaces $C_1(\LatticeRest{\rbs})$ and $\Im(\identity - \pi_0^{(\rbs)})$, respectively.
The boundary map for 1-cells,
\begin{align*}
\partial_1^{\LatticeMono{\rbs}}\qty(\qty{v_\rbs,v_\gbs,v_\bbs}) &= v_\rbs + \qty{v_\gbs,v_\bbs},
\end{align*}
can be separated into two functions $p_\mr{vert}^{(\rbs)}$ and $p_\mr{edge}^{(\rbs)}$ as
\begin{subequations}
\begin{align}
    p_\mr{vert}^{(\rbs)} &= P_{\Im(\identity - \pi_0^{(\rbs)})}\partial_1^{\LatticeMono{\rbs}}, \\
    p_\mr{edge}^{(\rbs)} &= P_{C_1(\LatticeRest{\rbs})}\partial_1^{\LatticeMono{\rbs}}, \\
    \partial_1^{\LatticeMono{\rbs}} &= p_\mr{vert}^{(\rbs)} + p_\mr{edge}^{(\rbs)}. \nonumber
\end{align}
\label{eq:projected_boundary_maps}%
\end{subequations}
These operators act on basis elements as
\begin{align*}
    p_\mr{vert}^{(\rbs)}(\qty{v_\rbs,v_\gbs,v_\bbs}) &= v_\rbs, \\
    p_\mr{edge}^{(\rbs)}(\qty{v_\rbs,v_\gbs,v_\bbs}) &= \qty{v_\gbs,v_\bbs}.
\end{align*}

All the above notations are naturally extended to other two colors \gbs and \bbs.

The decoding problem can be reformulated as follows: Given a set $\sigma$ of syndrome vertices in $\mathcal{H}$, identify a set of hyper-edges whose boundary is equal to the syndrome vertices.

The concatenated decoder first finds a matching $b_\rbs \in C_1(\LatticeRest{\rbs})$ in the red-restricted lattice such that $\partial_1^{\mathcal{L}^*_{\neg \rbs}} (b_\rbs) = \pi^{(\rbs)}_0(\sigma)$.
It then lifts this matching and the red vertices from the original syndrome into the red-only lattice to form a set of vertices $(\identity - \pi_0^{(\rbs)})(\sigma) + b_\rbs$, for which it then finds a matching $x$ such that $\partial_1^{\LatticeMono{\rbs}}(x) = (\identity - \pi_0^{(\rbs)})(\sigma) + b_\rbs$. This is repeated for the other two colors and the least-weight correction is selected as the final outcome, but we here only consider the red sub-decoding procedure (without loss of generality) as our goal is showing the validity of the decoder, not its optimality.

\subsection{Proof of validity when no boundaries}

We first show that, for a given syndrome and error, if each of the first- and second-round decoding processes works correctly (i.e. returns a valid matching), the concatenated MWPM decoder returns a valid matching.
Namely, we prove the following claim:

\begin{claim}
    For $\sigma \in C_0 (\mathcal{H})$, $x \in C_1 (\mathcal{H})$, and $b_\rbs \in C_1(\LatticeRest{\rbs})$, if
    \begin{subequations}
    \begin{align}
        \partial_1^{\LatticeRest{\rbs}}\qty(b_\rbs) &= \pi_0^{(\rbs)}(\sigma), \label{eq:validity_MWPM_step1} \\
        \partial_1^{\LatticeMono{\rbs}}(x) &= \qty(\identity - \pi_0^{(\rbs)})(\sigma) + b_\rbs, \label{eq:validity_MWPM_step2}
    \end{align}
    \label{eq:validity_MWPM}%
    \end{subequations}
    then 
    \begin{align}
        \partial_1^{\mathcal{H}}(x) = \sigma.
        \label{eq:validity_condition}
    \end{align}
\end{claim}

\begin{proof}
    Using the projected boundary maps in Eq.~\eqref{eq:projected_boundary_maps}, we rewrite Eq.~\eqref{eq:validity_MWPM_step2} as
    \begin{align*}
        p_\mr{vert}^{(\rbs)} (x) &= \qty(\identity - \pi_0^{(\rbs)})(\sigma), \\
        p_\mr{edge}^{(\rbs)} (x) &= b_\rbs.
    \end{align*}
    From the latter equation we get $\partial_1^{\LatticeRest{\rbs}} p_\mr{edge}^{(\rbs)} (x) = \pi_0^{(\rbs)} (\sigma)$, and thus $(\partial_1^{\LatticeRest{\rbs}} p_\mr{edge}^{(\rbs)} + p_\mr{vert}^{(\rbs)}) (x) = \sigma$.
    It suffices to prove that $(\partial_1^{\LatticeRest{\rbs}}  p_\mr{edge}^{(\rbs)} + p_\mr{vert}^{(\rbs)})(x) = \partial_1^{\mathcal{H}}(x)$. 
    Since its both sides are linear in $x$, we only need to show this over the basis of hyper-edges that span $C_1(\mathcal{H})$:
    \begin{align*}
        \qty(p_\mr{vert}^{(\rbs)} + \partial_1^{\LatticeRest{\rbs}} p_\mr{edge}^{(\rbs)})\qty(\qty{v_\rbs,v_\gbs,v_\bbs}) &= v_\rbs + \partial_1^{\LatticeRest{\rbs}} \qty(\qty{v_\gbs,v_\bbs}) \\
        &= v_\rbs + v_\gbs + v_\bbs \\
        &= \partial_1^\mathcal{H}\qty(\qty{v_\rbs,v_\gbs,v_\bbs}).
    \end{align*}
\end{proof}

\subsection{Consideration of boundaries \label{subapp:boundary_consideration}}

We now assume that $\mathcal{L}$ is a color-code lattice with boundaries.
To handle them, its dual lattice $\mathcal{L}^*$ is augmented with three distinctly colored boundary vertices $\BdryVert{\rbs}$, $\BdryVert{\gbs}$, and $\BdryVert{\bbs}$.
Vertices in $\mathcal{L}^*$ corresponding to green and blue faces of $\mathcal{L}$ adjacent to the red boundary are connected to $\BdryVert{\rbs}$, and so on for the other two colors, and all pairs of boundary vertices are connected.

Despite the existence of the boundary vertices, almost all the notations and definitions in Appendix~\ref{subsec:validity_proof_notations} do not change.
Denoting the set of boundary vertices of a lattice $\mathcal{L}$ as $V_\mr{bdry}(\mathcal{L})$, the boundary vertices of various lattices described above are set as follows:
\begin{align*}
    V_\mr{bdry}\qty(\mathcal{L}^*) = V_\mr{bdry}\qty(\mathcal{H}) &= \qty{\BdryVert{\rbs}, \BdryVert{\gbs}, \BdryVert{\bbs}}, \\
    V_\mr{bdry}\qty(\LatticeRest{\rbs}) &= \qty{\BdryVert{\gbs}, \BdryVert{\bbs}},  \\
    V_\mr{bdry}\qty(\LatticeMono{\rbs}) &= \qty{\BdryVert{\rbs}, \qty{\BdryVert{\gbs}, \BdryVert{\bbs}}},.
\end{align*}
We define $C_\mr{bdry}\qty(\mathcal{L})$ as the $\mathbbm{F}_2$ vector space spanned by $V_\mr{bdry}\qty(\mathcal{L})$.
Note that, from a decoding perspective, it does not matter if we contract multiple boundary vertices into a single boundary vertex since edges connecting them are given zero weight when applying MWPM; we thus displayed only one boundary vertex in Figs.~\ref{fig:color_code} and~\ref{fig:decoder_2D} for simplicity.
However, we here do not merge distinct boundary vertices for mathematical clarity.

Since syndromes can be matched with the boundary vertices, the mathematical description of MWPM in Eq.~\eqref{eq:validity_MWPM} and the validity condition of decoding in Eq.~\eqref{eq:validity_condition} should be modified.
The modified claim and its proof are as follows:
\begin{claim}
    For $\sigma \in C_0 (\mathcal{H})$, $x \in C_1 (\mathcal{H})$, and $b_\rbs \in C_1(\LatticeRest{\rbs})$, if 
    \begin{subequations}
    \begin{align}
        &\partial_1^{\LatticeRest{\rbs}}\qty(b_\rbs) - \pi_0^{(\rbs)}(\sigma) \in C_\mr{bdry}\qty(\LatticeRest{\rbs}), \label{eq:validity_MWPM_bdry_step1}\\
        &\partial_1^{\LatticeMono{\rbs}}(x) - \qty(\identity - \pi_0^{(\rbs)})(\sigma) - b_\rbs \in C_\mr{bdry}\qty(\LatticeMono{\rbs}), \label{eq:validity_MWPM_bdry_step2}
    \end{align}
    \label{eq:validity_MWPM_bdry}
    \end{subequations}
    then $\partial_1^{\mathcal{H}}(x) - \sigma \in C_\mr{bdry}(\mathcal{H})$.
\end{claim}

\begin{proof}
    Let $c \in C_\mr{bdry}(\LatticeRest{\rbs})$ and $c' = c'_\mr{vert} + c'_\mr{edge} \in C_\mr{bdry}(\LatticeMono{\rbs})$ denote the left-hand sides of Eqs.~\eqref{eq:validity_MWPM_bdry_step1} and~\eqref{eq:validity_MWPM_bdry_step2}, respectively, where $c'_\mr{vert} = P_{\Im (\identity - \pi_0^{(\rbs)})} c' \in \qty{0, \BdryVert{\rbs}}$ and $c'_\mr{edge} = P_{C_1(\LatticeRest{\rbs})} c' \in \qty{0, \qty{\BdryVert{\gbs}, \BdryVert{\bbs}}}$.
    We rewrite Eq.~\eqref{eq:validity_MWPM_bdry_step2} as
    \begin{align*}
        p_\mr{vert}^{(\rbs)} (x) &= \qty(\identity - \pi_0^{(\rbs)})(\sigma) + c'_\mr{vert}, \\
        p_\mr{edge}^{(\rbs)} (x) &= b_\rbs + c'_\mr{edge}.
    \end{align*}
    From the latter equation we get $\partial_1^{\LatticeRest{\rbs}} p_\mr{edge}^{(\rbs)} (x) = \pi_0^{(\rbs)} (\sigma) + c + \partial_1^{\LatticeRest{\rbs}} c'_\mr{edge}$, and thus $(\partial_1^{\LatticeRest{\rbs}} p_\mr{edge}^{(\rbs)} + p_\mr{vert}^{(\rbs)}) (x) = \sigma + c + c'_\mr{vert} + \partial_1^{\LatticeRest{\rbs}} c'_\mr{edge}$.
    Since $c + c'_\mr{vert} + \partial_1^{\LatticeRest{\rbs}} c'_\mr{edge} \in C_\mr{bdry}(\mathcal{H})$ (as $\partial_1^{\LatticeRest{\rbs}} c'_\mr{edge}$ is either 0 or $\BdryVert{\gbs} + \BdryVert{\bbs}$), it suffices to prove that
    \begin{align*}
        (\partial_1^{\LatticeRest{\rbs}}  p_\mr{edge}^{(\rbs)} + p_\mr{vert}^{(\rbs)} - \partial_1^{\mathcal{H}})(x) \in C_\mr{bdry}\qty(\mathcal{H}).
    \end{align*}
    Because of linearity, we only need to show this over the basis of hyper-edges that span $C_1(\mathcal{H})$:
    \begin{align*}
        &\qty(p_\mr{vert}^{(\rbs)} + \partial_1^{\LatticeRest{\rbs}} p_\mr{edge}^{(\rbs)} - \partial_1^{\mathcal{H}})\qty(\qty{v_\rbs,v_\gbs,v_\bbs}) \\ 
        ={}& v_\rbs + \partial_1^{\LatticeRest{\rbs}} \qty(\qty{v_\gbs,v_\bbs}) - \qty(v_\rbs + v_\gbs + v_\bbs) \\
        ={}& 0 \in C_\mr{bdry}\qty(\mathcal{H}).
    \end{align*}
\end{proof}

\subsection{Superiority of the concatenated MWPM decoder over the projection decoder}

In addition, we show that the projection decoder cannot outperform the concatenated MWPM decoder.
In other words, we analytically prove that the matching obtained from the concatenated MWPM decoder always has a weight not greater than that from the projection decoder.
We here only consider the cases without boundaries.

Given a syndrome $s$, the projection decoder finds matchings $b_\rbs$, $b_\gbs$, and $b_\bbs$ in the restricted lattices such that $\partial_1^{\mathcal{L}^*_{\neg \cbs}} (b_\cbs) = \pi^{(\cbs)}_0(\sigma)$ for each $\cbs \in \qty{\rbs, \gbs, \bbs}$. 
It then lifts the three restricted lattice matchings to find a hypergraph matching $x \in C_1(\mathcal{H})$ such that $\partial_2^*(x) = b_\rbs + b_\gbs + b_\bbs$. 

The following claim asserts that, for the projection and concatenated MWPM decoders sharing the same result on the red-restricted lattice, any matching on the projection decoder must also be a proper matching on the red-only lattice of the concatenated MWPM decoder. 
Hence, running MWPM on the red-only lattice always gives a matching whose weight is equal to or less than the wieght of any matching obtained from the projection decoder.

\begin{claim}
    Given $\sigma \in C_0(\mathcal{H})$, $x \in C_1(\mathcal{H})$, and $b_\cbs \in C_1(\LatticeRest{\cbs})$ for each $\cbs \in \qty{\rbs, \gbs, \bbs}$, suppose
    \begin{align*}
        \partial_1^{\LatticeRest{\cbs}} \qty(b_\cbs) &= \pi_0^{(\cbs)}(\sigma)\quad \forall \cbs \in \qty{\rbs, \gbs, \bbs}, \\
        \partial_2^* (x) &= b_\rbs + b_\gbs + b_\bbs, \\
        \partial_1^\mathcal{H} (x) &= \sigma.
    \end{align*}
    Then $\partial_1^{\LatticeMono{\rbs}}(x) = b_\rbs + \qty(\identity - \pi_0^{(\rbs)}) (\sigma)$.
\end{claim}

\begin{proof}
    The assertion is equivalent to
    \begin{align*}
        p_\mr{vert}^{(\rbs)} (x) &= \qty(\identity - \pi_0^{(\rbs)})(\sigma) \\
        p_\mr{edge}^{(\rbs)} (x) &= b_\rbs
    \end{align*}
    As $\partial_2^*(\qty{v_\rbs,v_\gbs,v_\bbs}) = \qty{v_\rbs,v_\gbs}+\qty{v_\gbs,v_\bbs}+\qty{v_\rbs,v_\bbs} = (p_\mr{edge}^{(\rbs)} + p_\mr{edge}^{(\gbs)} + p_\mr{edge}^{(\bbs)})(\qty{v_\rbs,v_\gbs,v_\bbs})$, applying linearity we get that $\partial_2^* (x) = (p_\mr{edge}^{(\rbs)} + p_\mr{edge}^{(\gbs)} + p_\mr{edge}^{(\bbs)})(x) = b_\rbs + b_\gbs + b_\bbs$. As the subspaces of red, green, and blue edges are disjoint, $p_\mr{edge}^{(\rbs)}(x) = b_\rbs$. 
    Similarly, $\partial_1^\mathcal{H}(\qty{v_\rbs,v_\gbs,v_\bbs}) = v_\rbs + v_\gbs + v_\bbs = (p_\mr{vert}^{(\rbs)} + p_\mr{vert}^{(\gbs)} + p_\mr{vert}^{(\bbs)})(\qty{v_\rbs,v_\gbs,v_\bbs})$. 
    By linearity, $\partial_1^\mathcal{H}(x) = (p_\mr{vert}^{(\rbs)} + p_\mr{vert}^{(\gbs)} + p_\mr{vert}^{(\bbs)})(x) = \sigma = [(\identity - \pi_0^{(\rbs)}) + (\identity - \pi_0^{(\gbs)}) + (\identity - \pi_0^{(\bbs)})](\sigma)$ as $\sum_\cbs (\identity - \pi_0^{(\cbs)}) = \identity$. 
    Again, as the subspaces of red, green, and blue vertices are disjoint, $p_\mr{vert}^{(\rbs)}(x) = (\identity - \pi_0^{(\rbs)})(\sigma)$.
\end{proof}

\section{Discussion on the origin of the bias between $\lgx$ and $\lgz$ failures \label{app:bias_origin}}

In this appendix, we discuss the origin of the bias between $\lgx$ and $\lgz$ failures and explain why the degree of this bias varies with the \cnot schedule.
Errors occurring during syndrome extraction propagate to other data qubits via \cnot gates, and the propagation pattern strongly depends on the \cnot schedule.
As a result, it is not surprising that failure rates are influenced by the schedule.
In particular, hook errors, which are errors on ancillary qubits that occur midway through syndrome extraction and propagate to more than one data qubit, can significantly impact the decoder’s performance, as they may reduce the fault distance.

To understand the origin of the bias, it is helpful to compare the optimal schedule,
\begin{align}
    \mathcal{A}_\mathrm{optimal} = [2, 3, 6, 5, 4, 1; 3, 4, 7, 6, 5, 2], \label{eq:highest_bias_schedule}
\end{align}
which exhibits a very small bias (discussed in Sec.~\ref{subsec:optimizing_cnot_schedule}), with the schedule that produces the highest bias,
\begin{align*}
    \mathcal{A}_\mathrm{biased} = [1, 6, 7, 5, 4, 2; 2, 3, 6, 7, 5, 4],
\end{align*}
where $R_\mathrm{bias} \approx 0.99$ and $p_\mathrm{fail}^{(X)} \approx 3 p_\mathrm{fail}^{(Z)}$ for $d=T=7$, $p=10^{-3}$.
The optimal schedule has the property that its $Z$-part and $X$-part follow the same pattern, differing consistently by 1.
This symmetry ensures that hook errors affect data qubits in the same way for the two ancillary qubits on a single face.
In contrast, the schedule $\mathcal{A}_\mathrm{biased}$ lacks such consistency; its $Z$- and $X$-parts are entirely different, resulting in hook errors exhibiting distinct characteristics for the two Pauli types.

\begin{figure}[!h]
	\centering
	\includegraphics[width=0.4\linewidth]{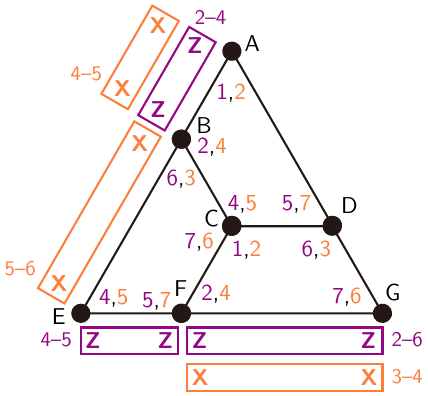}
	\caption{
        \textbf{Possible hook errors on the distance-3 color code patch with the CNOT schedule in Eq.~\eqref{eq:highest_bias_schedule}, producing the highest bias.}
        Each purple (orange) box labeled `$x$–$y$' adjacent to an edge between two qubits represents a $Z \otimes Z$ ($X \otimes X$) error on these qubits, which is caused by a Pauli-$Z$ ($X$) hook error occurring between time steps~$x$ and~$y$ on the nearby $Z$-type ($X$-type) ancillary qubit.
    }
	\label{fig:hook_errors}
\end{figure}

To illustrate this effect explicitly, Fig.~\ref{fig:hook_errors} shows a distance-3 patch for schedule $\mathcal{A}_\mathrm{biased}$, with possible hook errors visualized using purple and orange boxes.
For instance, the purple box adjacent to an edge connecting qubits~A and B (labeled `2--4' with `Z Z' inside) represents a $Z_\mathrm{A} \otimes Z_\mathrm{B}$ error caused by a hook error between time steps~2 and~4 on the nearby $Z$-type ancillary qubit (located on face A-B-C-D).
It is evident that $Z$-type hook errors are more likely to occur than $X$-type hook errors, as they involve a greater number of time steps (e.g., $Z_\mr{F} \otimes Z_\mr{G}$ involves time steps~2 to~6).
Consequently, $\lgz$ errors are expected to be more frequent than $\lgx$ errors, which aligns with the numerical results.
Generalizing this argument to higher code distances might be slightly more complex due to weight-6 checks, but the basic principles remain the same.

\section{Discussion on small-weight uncorrectable errors \label{app:least_weight_uncorrectable_error}}

In this appendix, we discuss errors with weights $< O(d/2)$ that cannot be corrected by the concatenated MWPM decoder.
We first consider weight-$O(d/3)$ errors, which may be difficult to correct via the projection decoder.
We verify that a specific family of weight-$O(d/3)$ errors that are uncorrectable via the projection decoder are correctable via the concatenated MWPM decoder.
We then show that there exist weight-$O(3d/7)$ errors that cannot be corrected by the concatenated MWPM decoder, which implies that the value of $\beta$ in the ansatz of Eq.~\eqref{eq:pfail_ansatz} may converge to $3/7$ for sufficiently large values of $d$.

\subsection{Small-weight errors uncorrectable via the projection decoder but correctable via the concatenated MWPM decoder \label{subsec:uncorrectable_errors_projection}}

\begin{figure*}[!t]
    \centering
    \includegraphics[width=0.9\textwidth]{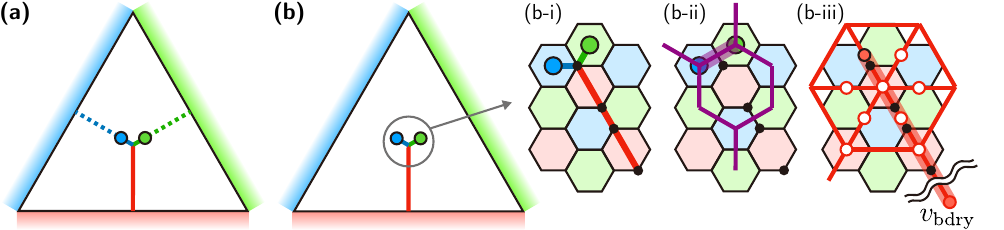}
    \caption{
    \textbf{Example of an error with weight $O(d/3)$ that is uncorrectable via the projection decoder but correctable via the concatenated MWPM decoder.}
    \subfig{a} The error is drawn as red, green, and blue solid lines, which makes a blue check and a green check marked as circles violated.
    The projection decoder pairs up these checks with the green/blue boundaries (dotted lines), thus it makes a logical error.
    \subfig{b} The same error is correctable via the concatenated MWPM decoder. 
    The microscopic picture of the vicinity of the violated checks is described in (b-i), where black dots indicate qubits with errors.
    In (b-ii) and (b-iii), MWPMs on the red-restricted and red-only lattices are schematized with the notations used in Fig.~\ref{fig:decoder_2D}, where only some of the lattice elements are visualized for simplicity.
    The outcome of the second-round MWPM, which is visualized as a red thick line ending at the boundary vertex $v_\mr{bdry}$, coincides with the error, thus it makes no logical error.
    }
    \label{fig:uncorrectable_error_projection_decoder}
\end{figure*}

In Fig.~\ref{fig:uncorrectable_error_projection_decoder}(a), we schematize an error with weight $O(d/3)$ that is uncorrectable via the projection decoder \cite{chamberland2020triangular,beverland2021cost,sahay2022decoder}.
The error is composed of a single-qubit error at the center of the triangular patch and a red string operator connecting the center and the red boundary, which is drawn as red, green, and blue solid lines.
This error makes a blue check and a green check (marked as circles) violated.
In the green-restricted (blue-restricted) lattice, the green (blue) check is paired up with the green (blue) boundary, thus the decoder makes a wrong prediction that causes a logical error.
See Appendix~A of Ref.~\cite{chamberland2020triangular} for an explicit example of such an error.

Notably, the above error is correctable via the concatenated MWPM decoder, as described in Fig.~\ref{fig:uncorrectable_error_projection_decoder}(b).
Roughly speaking, it is thanks to the strategy to select the least-weight one among the three predictions $\qty{\wtilde{V}_\mr{pred}^{(\rbs)}, \wtilde{V}_\mr{pred}^{(\gbs)}, \wtilde{V}_\mr{pred}^{(\bbs)}}$, which are obtained from the sub-decoding procedures of red, green, and blue, respectively.
Namely, MWPM on the red-restricted and red-only lattices, which are microscopically described in (b-ii) and (b-iii), gives a correct prediction $\wtilde{V}_\mr{pred}^{(\rbs)}$.
Although other two predictions $\wtilde{V}_\mr{pred}^{(\gbs)}$ and $\wtilde{V}_\mr{pred}^{(\bbs)}$ make logical errors, their weights are $O(2d/3)$, thus $\wtilde{V}_\mr{pred}^{(\rbs)}$ is selected as the final outcome of the decoder.

\subsection{Small-weight errors uncorrectable via the concatenated MWPM decoder \label{subsec:uncorrectable_errors_concat}}

\begin{figure*}[!t]
    \centering
    \includegraphics[width=0.9\textwidth]{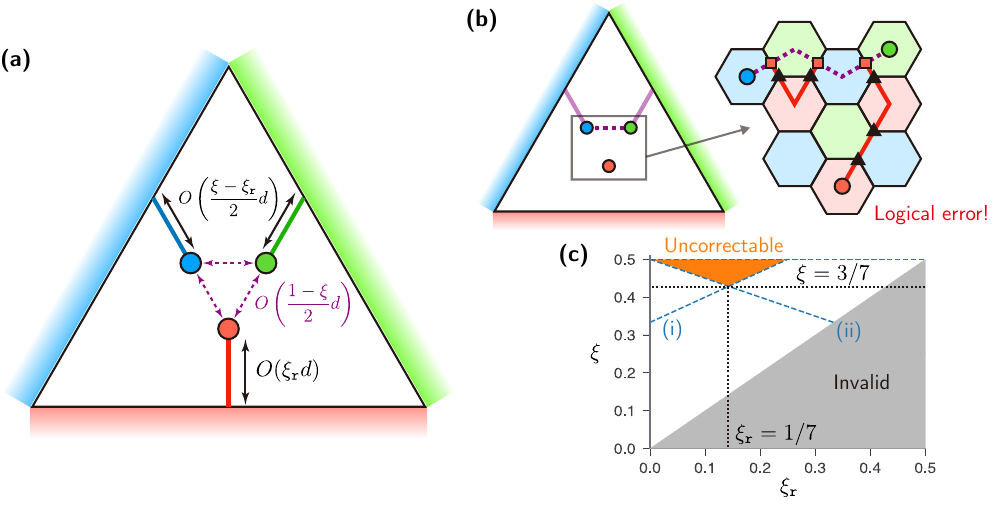}
    \caption{
        \textbf{Example of an error with weight $<d/2$ that may be uncorrectable via the concatenated MWPM decoder.}
        \subfig{a} The error has a weight of $w = O(\xi d)$ for $\xi < 1/2$ and is composed of three string operators (one for each color) terminating at the three boundaries, whose weights are respectively $w_\rbs = O(\xi_\rbs d)$, $w_\gbs = O([(\xi - \xi_\rbs)/2]d)$, and $w_\bbs = w - w_\rbs - w_\gbs = O([(\xi - \xi_\rbs)/2]d)$, where $\xi_\rbs \leq \xi$.
        They are arranged in a way that every pair of violated checks (say, green and blue) is separated by $w_\mr{gap} = O([(1-\xi)/2]d)$ edges of the third color (red).
        \subfig{b} If Eq.~\eqref{eq:uncorrectable_condition_red} holds, the first-round MWPM on the red-restricted lattice returns a wrong matching, which leads to a logical error.
        As an example, the microscopic picture near the violated checks is illustrated for the case of $w_\mr{gap} = 3$, where the purple dotted line indicates the matching of the first-round MWPM, the red squares are the corresponding red edges $E_\mr{pred}^{(\rbs)}$, the red solid lines are the matching of the second-round MWPM, and the black triangles indicate the final prediction.
        Similarly, decoding on green or blue sub-lattices incurs a logical error if Eq.~\eqref{eq:uncorrectable_condition_green} or~\eqref{eq:uncorrectable_condition_blue} holds.
        The decoding eventually fails only when all of Eqs.~\eqref{eq:uncorrectable_condition_red}--\eqref{eq:uncorrectable_condition_blue} are satisfied.
        \subfig{c} The region of $\xi$ and $\xi_\rbs$ that give uncorrectable errors is visualized.
        The blue dashed lines (i) and (ii) respectively indicate the boundaries of the regions where Eqs.~\eqref{eq:uncorrectable_condition_red} and \eqref{eq:uncorrectable_condition_green}/\eqref{eq:uncorrectable_condition_blue} hold.
        The gray region corresponds to $\xi_\rbs > \xi$, which is invalid.
    }
    \label{fig:uncorrectable_error_concatenated_decoder}
\end{figure*}

Alike the Möbius MWPM decoder \cite{sahay2022decoder}, the concatenated MWPM decoder fails to overcome the limitation that $O(3d/7)$-weight errors may not be correctable.
We consider an error with weight $w = O(\xi d)$ for $\xi < 1/2$ as presented in Fig.~\ref{fig:uncorrectable_error_concatenated_decoder}(a), which is composed of three string operators (one for each color) that terminate at the corresponding boundaries and make one check violated for each color.
These string operators have the weights of $w_\rbs = O(\xi_\rbs d)$, $w_\gbs = O([(\xi - \xi_\rbs)/2]d)$, and $w_\bbs = w - w_\rbs - w_\gbs = O([(\xi - \xi_\rbs)/2]d)$, where $\xi_\rbs \leq \xi$.
Additionally, they are arranged in a way that each pair of violated checks (say, green and blue) are separated by $w_\mr{gap} = O([(1-\xi)/2]d)$ edges of the third color (red); that is, to move from the violated green check to the blue check, we need to pass at least $w_\mr{gap}$ red edges.
(Note that a non-trivial string-net operator constructed by moving the violated blue and green checks to the violated red check and fusing them has a weight of $w + 2w_\mr{gap} = O(d)$.)
When this error occurs, we face a problem that a first-round MWPM of a sub-decoding procedure can be wrong.
For example, as shown in Fig.~\ref{fig:uncorrectable_error_concatenated_decoder}(b), the MWPM on the red-restricted lattice, which involves only the violated green and blue checks, gives a wrong matching connecting these checks (drawn as a dotted purple line) if and only if the distance between them ($w_\mr{gap}$) is smaller than the sum of their distances to the corresponding boundaries ($w_\gbs + w_\bbs$) within the red-restricted lattice, i.e.,
\begin{align}
    w_\mr{gap} = O\qty(\frac{1-\xi}{2}d) < w_\gbs + w_\bbs = O\qty(\qty(\xi - \xi_\rbs)d).
    \label{eq:uncorrectable_condition_red}
\end{align}
Then the subsequent MWPM on the red-only lattice produces predictions represented by black triangles in Fig.~\ref{fig:uncorrectable_error_concatenated_decoder}(b), leading to a logical error.
Similarly, the green sub-decoding fails if and only if
\begin{align}
    w_\mr{gap} = O\qty(\frac{1-\xi}{2}d) < w_\rbs + w_\bbs = O\qty(\frac{\xi + \xi_\rbs}{2}d),
    \label{eq:uncorrectable_condition_green}
\end{align}
and lastly, the blue sub-decoding fails if and only if
\begin{align}
    w_\mr{gap} = O\qty(\frac{1-\xi}{2}d) < w_\rbs + w_\gbs = O\qty(\frac{\xi + \xi_\rbs}{2}d).
    \label{eq:uncorrectable_condition_blue}
\end{align}
Since we select the least-weight correction among the three colors, the decoding eventually fails only when all of Eqs.~\eqref{eq:uncorrectable_condition_red}--\eqref{eq:uncorrectable_condition_blue} hold.
The region of such uncorrectable errors is visualized in Fig.~\ref{fig:uncorrectable_error_concatenated_decoder}(c) on the plane with the axes of $\xi_\rbs$ and $\xi$, while neglecting constant factors.
The minimum of $\xi$ in this region is $3/7$, which corresponds to $\xi_\rbs = 1/7$; thus, there may exist errors with weights $\geq O(3d/7)$ but smaller than $d/2$ that are uncorrectable via the concatenated MWPM decoder.

\begin{figure*}[!t]
	\centering
	\includegraphics[width=0.9\textwidth]{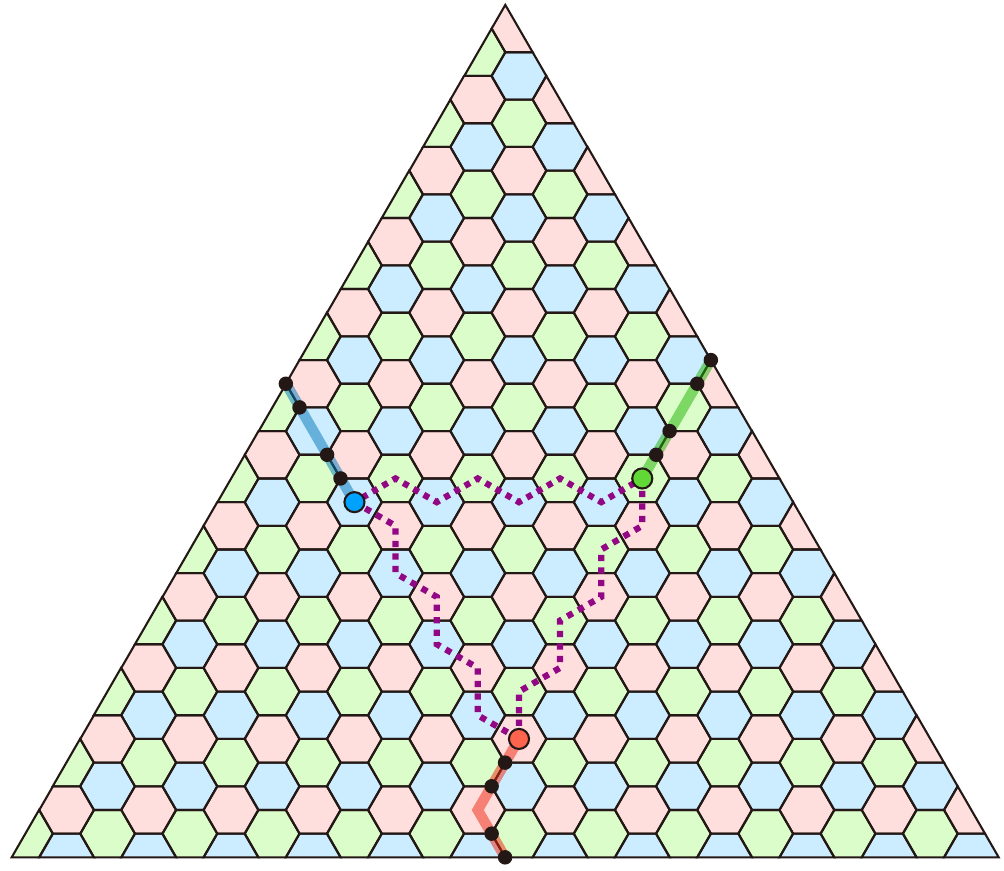}
	\caption{
     \textbf{Smallest example of an uncorrectable error shown in Fig.~\ref{fig:uncorrectable_error_concatenated_decoder}.}
     The error, marked as black dots, lies on the code of $d=25$, which is the smallest possible code distance permitting such errors, and has the weight of $w = 12 < d/2$.
     Three string operators (colored thick lines) constituting the error have the same weight of $w_\rbs = w_\gbs = w_\bbs = 4$ and each pair of violated checks (colored circles) is separated by $w_\mr{gap} = 7$ edges.
     The concatenated MWPM decoder cannot correct this error since it satisfies all of Eqs.~\eqref{eq:uncorrectable_condition_red}--\eqref{eq:uncorrectable_condition_blue}.
 }
	\label{fig:uncorrectable_error_example}
\end{figure*}

We exhaustively search for errors of the above type while increasing $d$ under the following constraints:
(i) The red string operator terminates exactly at the center qubit of the red boundary, that is, $(d - 1)/2$ qubits are placed on each of the left and right sides of the center qubit along the red boundary.
(ii) The weights of the blue and green string operators differ by up to two.
(iii) The value of $w_\mr{gap}$ is the same for every pair of the violated checks.
We find that $d=25$ is the smallest code distance that permits such uncorrectable errors, as shown in the example visualized in Fig.~\ref{fig:uncorrectable_error_example}, where $d=25$, $w = 12$, $w_\rbs = w_\gbs = w_\bbs = 4$, and $w_\mr{gap} = 7$.

\section{Additional numerical analyses \label{app:more_numerical_analyses}}

In this appendix, we present the results of additional numerical analyses.

\begin{figure}[!t]
	\centering
	\includegraphics[width=0.7\linewidth]{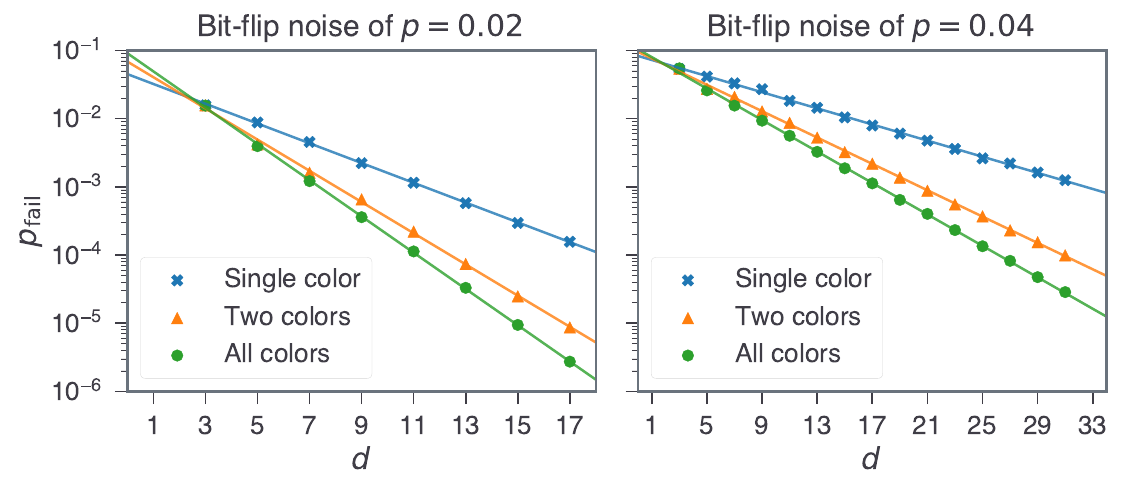}
	\caption{
    \textbf{Comparison of different color-selecting strategies.}
    We consider three strategies for the color-selecting step of the concatenated MWPM decoder: using only the outcome of the sub-decoding procedure of a single color (\rbs) without a comparison, comparing the outcomes for two colors (\rbs, \gbs), and comparing the outcomes for all the three colors (which is the original strategy).
    The failure rates $p_\mr{fail}$ under the bit-flip noise models of $p = 0.02$ (left) and $p = 0.04$ (right) are plotted against $d$ for these three strategies.
    The solid lines indicate the linear regressions of $\log p_\mr{fail}$ against $d$: $\log p_\mr{fail} = -0.33d - 3.1$ ($p=0.02$, single color), $-0.53d - 2.7$ ($p=0.02$, two colors), $-0.61d - 2.4$ ($p=0.02$, all colors), $-0.14d - 2.5$ ($p=0.04$, single color), $-0.22d - 2.4$ ($p=0.04$, two colors), and $-0.27d - 2.3$ ($p=0.04$, all colors), respectively.
    }
	\label{fig:color_strategy_comparison}
\end{figure}

Figure~\ref{fig:color_strategy_comparison} compares our color-selecting strategy (i.e., executing the sub-decoding procedures for all the three colors and selecting the best one) with two other strategies of performing only one or two sub-decoding procedures.
It plots logical failure rates under the bit-flip noise model of $p \in \qty{0.02, 0.04}$ for these three strategies, confirming that our original strategy of using all the three colors significantly outperforms the others.

\begin{figure*}[!t]
	\centering
	\includegraphics[width=0.82\textwidth]{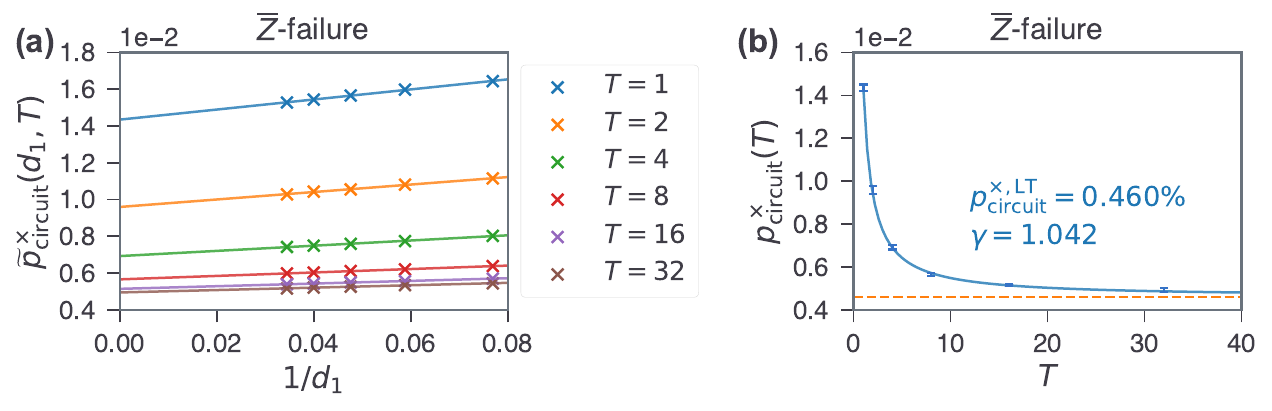}
     \includegraphics[width=0.82\textwidth]{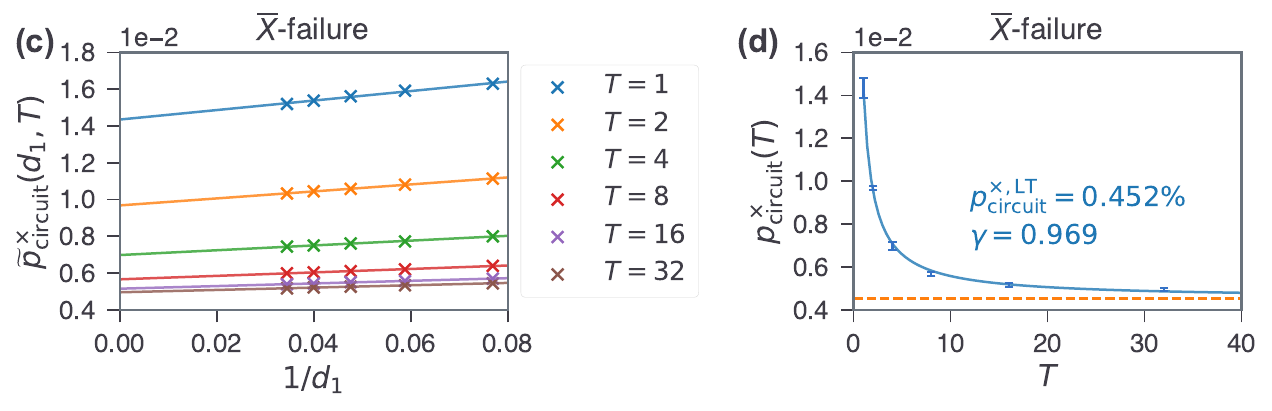}
	\caption{
    \textbf{Separate analyses for the $\ov{Z}$- and $\ov{X}$-failure of the decoder under near-threshold circuit-level noise.}
    For the $\ov{Z}$-failure, \subfig{a} and \subfig{b} are the counterparts of Figs.~\ref{fig:circuit_numerical_analysis_near_thrs}(b) and~(c), respectively.
    Similarly, \subfig{c} and \subfig{d} are their counterparts for the $\ov{X}$-failure.
    The long-term cross thresholds $\CrossThrsLT{circuit}$ are estimated as 0.460\% for the $\ov{Z}$-failure and 0.452\% for the $\ov{X}$-failure, which differ by about 1.8\%.
    }
	\label{fig:near_threshold_paulis}
\end{figure*}

Figure~\ref{fig:near_threshold_paulis} presents separate analyses for the $\ov{Z}$- and $\ov{X}$-failure of the decoder under near-threshold circuit-level noise, which are analogous to Figs.~\ref{fig:circuit_numerical_analysis_near_thrs}(b) and~(c).
The long-term cross thresholds $\CrossThrsLT{circuit}$ are estimated as 0.460\% for the $\ov{Z}$-failure and 0.452\% for the $\ov{X}$-failure, which differ by about 1.8\% (0.008\%p).
Note that $\CrossThrsLT{circuit}$ for the overall failure rate is about 0.455\%, which is the average of these two.

\begin{figure}[!t]
	\centering
	\includegraphics[width=0.48\linewidth]{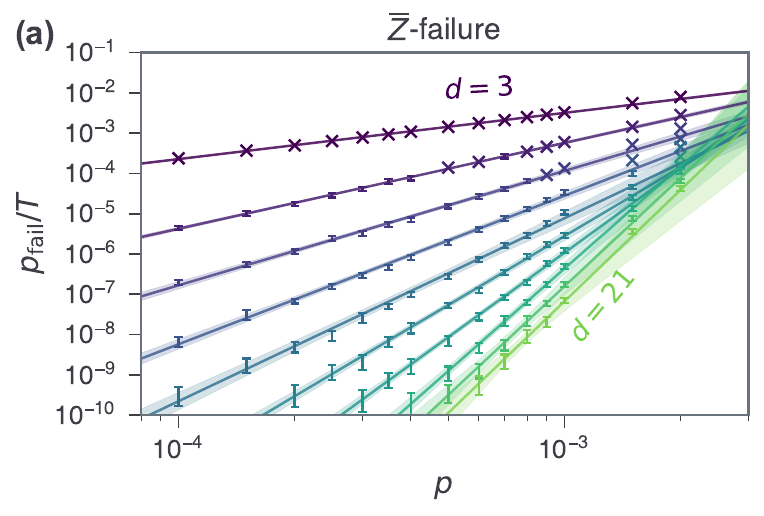}
    \includegraphics[width=0.48\linewidth]{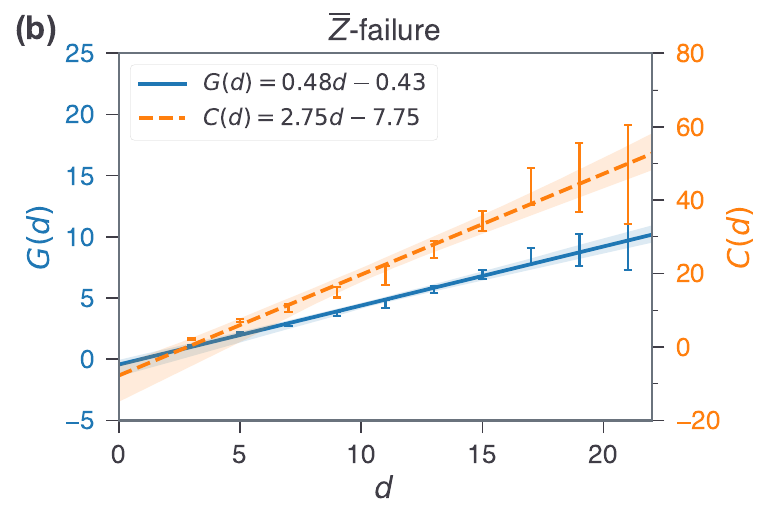}
    \includegraphics[width=0.48\linewidth]{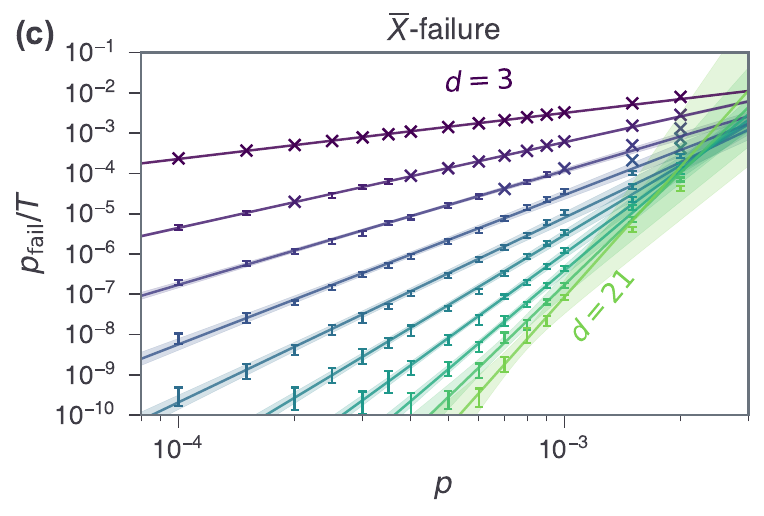}
    \includegraphics[width=0.48\linewidth]{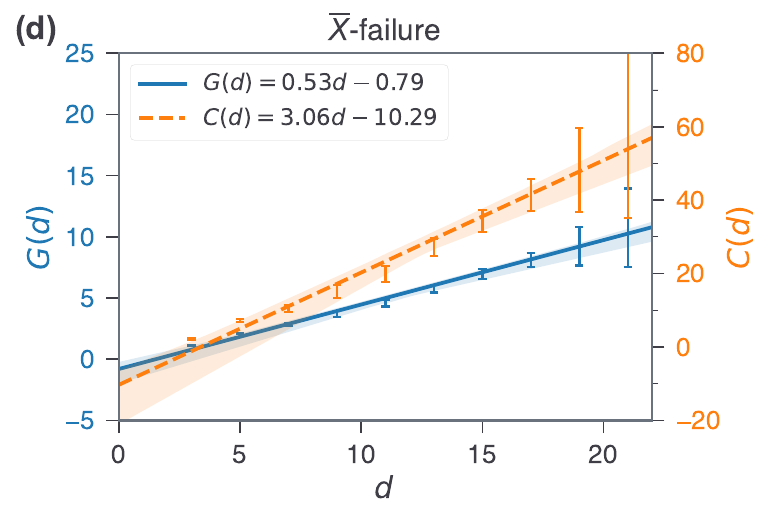}
	\caption{
    \textbf{Separate analyses for the $\ov{Z}$- and $\ov{X}$-failure of the decoder under sub-threshold circuit-level noise.}
    For the $\ov{Z}$-failure, \subfig{a} and \subfig{b} are the counterparts of Figs.~\ref{fig:pfails_low_errors_circuit}(a) and~(b), respectively.
    Similarly, \subfig{c} and \subfig{d} are their counterparts for the $\ov{X}$-failure.
    }
    \label{fig:sub_threshold_paulis}
\end{figure}

Figure~\ref{fig:sub_threshold_paulis} presents separate analyses for the $\ov{Z}$- and $\ov{X}$-failure of the decoder under sub-threshold circuit-level noise, which are analogous to Fig.~\ref{fig:pfails_low_errors_circuit}.
For the $\ov{Z}$-failure, the regressions of $G(d)$ and $C(d)$ against $d$ and the corresponding ansatz parameters are
\begin{align*}
    &G(d) = (0.48 \pm 0.07) (d - 12) + (5.4 \pm 0.4), \quad C(d) =(2.7 \pm 0.5) (d - 12) + (25 \pm 3), \\ 
    &\ScalingThrs{circuit} \approx 3.4 \times 10^{-3}, \quad \alpha \approx 5.0 \times 10^{-3}, \quad \beta \approx 0.48, \quad \eta \approx 5.4, \quad d_0 = 12.
\end{align*}
For the $\ov{X}$-failure, they are
\begin{align*}
    &G(d) = (0.53 \pm 0.05) (d - 12) + (5.5 \pm 0.3), \quad C(d) = (3.1 \pm 0.5) (d - 12) + (26 \pm 3), \\ 
    &\ScalingThrs{circuit} \approx 3.0 \times 10^{-3}, \quad \alpha \approx 3.4 \times 10^{-3}, \quad \beta \approx 0.53, \quad \eta \approx 5.5, \quad d_0 = 12.
\end{align*}

\begin{figure}[!t]
	\centering
	\includegraphics[width=0.8\linewidth]{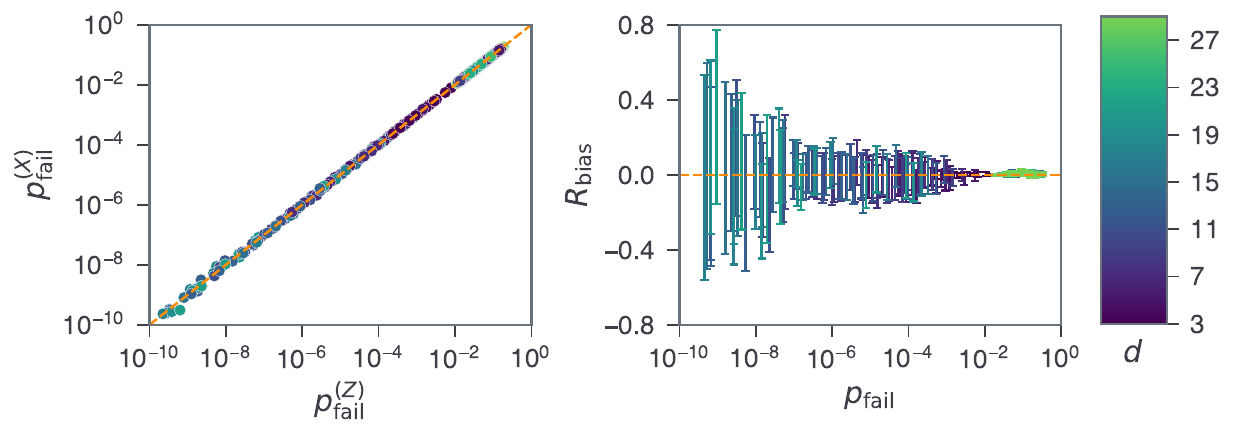}
	\caption{
    \textbf{Analysis of bias between $\ov{Z}$- and $\ov{X}$-failure under circuit-level noise.}
    The left figure shows the scatter plot of the $\ov{Z}$-failure rate $p_\mr{fail}^{(Z)}$ and the $\ov{X}$-failure rate $p_\mr{fail}^{(X)}$ for all the simulation outcomes in Figs.~\ref{fig:circuit_numerical_analysis_near_thrs} and~\ref{fig:pfails_low_errors_circuit}.
    The right figure shows the scatter plot of the logical failure rate $p_\mr{fail}$ and the bias $R_\mathrm{bias}$ defined in Eq.~\eqref{eq:bias_coeff} for all the simulation outcomes in Figs.~\ref{fig:circuit_numerical_analysis_near_thrs} and~\ref{fig:pfails_low_errors_circuit}.
    The error bars indicate the 99\% confidence intervals of $R_\mathrm{bias}$.
    For both figures, the dots and error bars are colored depending on the code distance $d$ (as shown in the right color bar) and the orange dashed lines represent unbiased cases of $p_\mr{fail}^{(Z)} = p_\mr{fail}^{(X)}$.
    }
	\label{fig:bias_scatter}
\end{figure}

Figure~\ref{fig:bias_scatter} analyzes the bias between $\ov{Z}$- and $\ov{X}$-failure in all the circuit-level simulation outcomes of Fig.~\ref{fig:circuit_numerical_analysis_near_thrs} and~\ref{fig:pfails_low_errors_circuit}.
It shows two scatter plots: one for the $\ov{Z}$-failure rate $p_\mr{fail}^{(Z)}$ and the $\ov{X}$-failure rate $p_\mr{fail}^{(X)}$ and the other one for the total failure rate $p_\mr{fail}$ and the bias $R_\mr{bias}$ defined in Eq.~\eqref{eq:bias_coeff}.
We observe that the 99\% confidence intervals of $R_\mr{bias}$ almost always include zero, which implies that we cannot reject the null hypothesis that there is no bias.


\end{document}